\documentclass[11pt,oneside,english]{article}

\usepackage{amsmath,amsfonts,amssymb}
\usepackage{authblk}
\usepackage{bm}
\usepackage{booktabs}
\usepackage[font=footnotesize]{caption}
\usepackage{cite}
\usepackage[T1]{fontenc}
\usepackage[verbose,tmargin=2.6cm,bmargin=2.6cm,lmargin=2.5cm,rmargin=2.5cm,footskip=1cm]{geometry}
\usepackage[utf8]{inputenc}
\usepackage{lmodern}
\usepackage{mathrsfs}
\usepackage{mathtools}
\usepackage{xcolor}
\usepackage[colorlinks=true,linkcolor={red!70!black},citecolor={blue!90},urlcolor={blue!80!black}]{hyperref}
\usepackage{tensor}
\date{}
\usepackage{amsthm}
\usepackage{tikz-cd}

\usepackage{mmacells}
\mmaSet{
  morefv={gobble=7}
}
\mmaDefineMathReplacement{��}{\mathbbm{e}}
\mmaDefineMathReplacement{`}{\text{\`{}}}

\usepackage{tikz}
\usetikzlibrary{arrows,decorations.markings}
\tikzset{
  big blue arrow/.style={
    decoration={markings,mark=at position 1 with {\arrow[scale=2,#1]{>}}},
    postaction={decorate},
    shorten >=0.4pt},
  big blue arrow/.default=blue}
  \tikzset{
  big red arrow/.style={
    decoration={markings,mark=at position 1 with {\arrow[scale=2,#1]{>}}},
    postaction={decorate},
    shorten >=0.4pt},
  big red arrow/.default=red}
    \tikzset{
  big arrow/.style={
    decoration={markings,mark=at position 1 with {\arrow[scale=2,#1]{>}}},
    postaction={decorate},
    shorten >=0.4pt},
  big arrow/.default=black}

\usepackage[nameinlink,capitalize]{cleveref}
\crefname{figure}{Figure}{Figures}
\labelcrefformat{subequation}{#2(#1)#3}
\labelcrefrangeformat{subequation}{#3(#1)#4 to #5(#2)#6}

\numberwithin{equation}{section}

\makeatletter
\renewenvironment{abstract}{%
    \if@twocolumn
      \section*{\abstractname}%
    \else 
      \begin{center}%
        {\bfseries \normalsize\abstractname\vspace{\z@}}
      \end{center}%
      \quotation
    \fi}
    {\if@twocolumn\else\endquotation\fi}
\makeatother

\makeatletter
\g@addto@macro\bfseries{\boldmath}

\let\originalleft\left
\let\originalright\right
\renewcommand*{\left}{\mathopen{}\mathclose\bgroup\originalleft}
\renewcommand*{\right}{\aftergroup\egroup\originalright}

\renewcommand*{\to}{\mathchoice{\longrightarrow}{\rightarrow}{\rightarrow}{\rightarrow}}

\newtheorem{theorem}{Theorem}[section]

\newtheorem{lemma}[theorem]{Lemma}
\newtheorem{corollary}[theorem]{Corollary}

\theoremstyle{definition}
\newtheorem{definition}[theorem]{Definition}
\newtheorem{remark}{Remark}[theorem]

\newcommand\G{\mathbb{G}}

\newcommand{\ra}{\rightarrow}

\newcommand\A{\mathbb{A}}

\definecolor{Vero}{RGB}{255,215,0}
\definecolor{SO}{RGB}{0,40,255}

\makeatletter
\newcommand*{\diff}{\@ifnextchar^{\DIFF}{\DIFF^{}}}
\def\DIFF^#1{\mathop{\mathrm{\mathstrut d}}\nolimits^{#1}\gobblespace}
\def\gobblespace{\futurelet\diffarg\opspace}
\def\opspace{%
    \let\DiffSpace\!%
    \ifx\diffarg(%
        \let\DiffSpace\relax
    \else
        \ifx\diffarg[%
            \let\DiffSpace\relax
        \else
            \ifx\diffarg\{%
                \let\DiffSpace\relax
            \fi
        \fi
    \fi
    \DiffSpace
}
\makeatother

\title{
\Huge Intersection Theory on Weighted Blowups of F-theory Vacua}
\author{\Large Veronica Arena\thanks{veronica\_arena@brown.edu}}
\affil{\normalsize \emph{Department of Mathematics, Brown University, 151 Thayer Street, Box 1917, Providence, RI 02912}}
\author{\Large Patrick Jefferson\thanks{pjeffers@mit.edu}}
\affil{\normalsize \emph{Center for Theoretical Physics, Department of Physics, Massachusetts Institute of Technology, 77 Massachusetts Avenue, Cambridge, MA 02139, USA}}
\author{\Large Stephen Obinna\thanks{stephen\_obinna@brown.edu}}
\affil{\normalsize \emph{Department of Mathematics, Brown University, 151 Thayer Street, Box 1917, Providence, RI 02912}}

\begin{document}

\maketitle
\begin{tikzpicture}[remember picture,overlay]
   \node[anchor=north east,inner sep=0pt] at (current page.north east)
              {$\begin{array}{ccc}&&\\ \\ \text{MIT-CTP-5551}&&\end{array}$};
\end{tikzpicture}
\thispagestyle{empty}
\begin{abstract}
\noindent Generalizing the results of \texttt{1211.6077} and \texttt{1703.00905}, we prove a formula for the pushforward of an arbitrary analytic function of the exceptional divisor class of a weighted blowup of an algebraic variety centered at a smooth complete intersection with normal crossing. We check this formula extensively by computing the generating function of intersection numbers of a weighted blowup of the generic SU(5) Tate model over arbitrary smooth base, and comparing the answer to known results. Motivated by applications to four-dimensional F-theory flux compactifications, we use our formula to compute the intersection pairing on the vertical part of the middle cohomology of elliptic Calabi-Yau 4-folds resolving the generic F$_4$ and Sp(6) Tate models with non-minimal singularities. These resolutions lead to non-flat fibrations in which certain fibers contain 3-fold (divisor) components, whose physical interpretation in M/F-theory remains to be fully explored. 
\end{abstract}
\flushbottom
\newpage
\tableofcontents
\addtocontents{toc}{\protect\thispagestyle{empty}}
\setcounter{page}{1}

\section{Introduction and Summary of Results}
\label{sec:intro}

String compactifications---vacuum solutions of string theory in which the extra spatial dimensions are curled into compact geometries---are central to efforts to connect string theory to the observable universe, as they in principle provide a means to explain why the extra dimensions of string theory have thus far remained experimentally undetected. A particularly tractable and well-controlled class of vacuum solutions is given by compactifications on Calabi-Yau (CY) $n$-folds, which satisfy the vacuum Einstein equations and preserve fractional supersymmetry in approximate low energy descriptions of string theory \cite{JoyceCompact}.\footnote{In principle, we define a CY $n$-fold to be a Ricci-flat compact K\"ahler manifold with complex dimension $n$ and global holonomy contained in SU$(n)$. In practice, however, we adopt a less restrictive definition and take a CY $n$-fold to be a projective variety with vanishing first Chern class. The vanishing of the first Chern class implies the existence of a finite cover with global holonomy contained in SU($n$) and is moreover a straightforward condition to implement mathematically. Note that various related, but distinct definitions of CY manifolds appear throughout the literature \cite{TianYau1,TianYau2,GHJ}.} 

Although smooth CY $n$-folds are challenging to construct explicitly, one of the insights to emerge from early work on string compactifications on toroidal orbifolds \cite{Dixon:1985jw} was that CY $n$-folds can often be constructed easily in singular limits---for example, a K3 surface (which is a CY 2-fold) can be constructed simply as an orbifold $T^4/\mathbb Z_2$ that is singular at 16 fixed points (see, e.g., \cite{Vafa:1997pm}). Thus, singular CY $n$-folds provide a convenient starting point for studying interesting (and what turns out to be rather large) classes of string vacua. However, there is a price to be paid for studying singular compactifications rather than smooth ones---for many applications, one must also identify a suitable means of ``resolving'' the singularities in order to reach a limit of the theory that admits a conventional low energy description in supergravity \cite{Aspinwall:1994ev}. This is because the standard procedure for determining the relevant kinematic properties of the low energy approximate description of a string compactification is Kaluza-Klein reduction, which is most easily performed by integrating fields over the compact internal space using classic techniques from complex geometry \cite{becker_becker_schwarz_2006}.

A useful mathematical setting in which to study singular CY $n$-folds is algebraic geometry, where one typically constructs a CY $n$-fold as a complete intersection of hypersurfaces in a projective variety. An advantage of constructing CY $n$-folds in this manner is that singularities of algebraic varieties can be systematically resolved by means of a sequence of transformations called `blowups'.\footnote{Roughly speaking, a blowup is a type of birational transformation that replaces a choice of subspace of an algebraic variety with the directions pointing out of that subspace, i.e., a projective space. A blowup of an algebraic variety $Y$ that replaces a subvariety $Z \subset Y$ with a projective space is said to be `centered' at $Z$. A more detailed discussion of blowups appears in \cref{sec:toricweighted}.} However, in order to ensure that the resulting geometry retains its physical interpretation as a string vacuum solution, any sequence of blowups must be chosen to ensure that the resulting space is also a CY $n$-fold; a blowup or sequence or blowups satisfying this condition is called `crepant' \cite{reid1983minimal}. This requirement places severe restrictions on the possible blowups that can be performed on singularities of CY in string compactitications.\footnote{For certain singularities, in fact, it has been argued that no finite sequence of blowups preserving the CY condition exists \cite{Arras:2016evy}.}

Although ordinary crepant blowups can in many cases be quite difficult to obtain, there fortunately exists a well-known and useful generalization called `weighted' blowups, in which each of point of a subspace of an algebraic variety is replaced, more generally, with a weighted projective space \cite{pub.1089196667,10.2307/2152704}. Weighted blowups have been used to identify crepant resolutions of numerous examples of CY singularities that would be difficult or perhaps impossible to resolve by means of ordinary blowups (see, e.g., \cite{Candelas:2000nc,Esole:2014hya}). For singular CY $n$-folds defined as complete intersections of hypersurfaces in projective toric varieties, weighted blowups can be implemented quite naturally and have been used extensively to compute topological properties of their resolutions, relevant to the physics of string compactifications on such spaces \cite{Braun:2014kla}. By comparison, available techniques for studying topological properties of weighted blowups of non-toric algebraic varieties have not been explored as extensively, and this has somewhat restricted the scope of singular string compactifications that can be usefully explored.

Let us elaborate on some of the topological properties of resolved CY that one may wish to study. For concreteness, let $X' \rightarrow X$ be the resolution of a singular CY $n$-fold $X$. One property of smooth CY $n$-folds relevant for understanding the physics of string compactifications are topological intersection numbers of codimension-one algebraic cycles, called `divisors'. Given a basis of divisors $\hat D_I \subset H_{n-1,n-1}(X')$, the topological intersection numbers of $X'$ can be expressed as integrals
	\begin{align}
		\int_{X'} \text{PD}(\hat D_{I_1}) \wedge \cdots \wedge \text{PD}( \hat D_{I_n}),
	\end{align}
where `PD' indicates the Poincar\'e dual and hence $\text{PD}(\hat D_I) \in H^{1,1}(X')$. Integrals of the above form determine, for example, the coupling constants of topological terms appearing in the supergravity action describing M-theory compactifications on CY $3$-folds and $4$-folds; see the introduction of \cite{Jefferson:2022xft} and references therein for a more detailed discussion. When $X' \subset Y'$ is a complete intersection $X' = X'_1 \cap \cdots \cap X'_k$ in the blowup of an algebraic variety, $Y' \rightarrow Y$, and moreover a suitable basis of divisors of $Y'$ is known (in terms of which the divisors of $X'$ can also be realized as complete intersections $\hat D_I =  D_I \cap  X'$), then the above intersection numbers can be expressed as 
	\begin{align}
	\label{eqn:intY}
			\int_{Y'} \text{PD}( D_{I_1}) \wedge \cdots \wedge \text{PD}( D_{I_d}) \wedge \text{PD}( X') = \boldsymbol D_{I_1} \cdots \boldsymbol D_{I_n}  \boldsymbol X',
	\end{align}
where on the right-hand side of the above equation, the intersection numbers are evaluated as products of the divisor classes $\boldsymbol D_I, \boldsymbol X'$ in the Chow ring of $Y'$, which is a ring that encodes the intersection properties of various subvarieties in $Y'$. Thus, the problem of computing the topological intersection numbers of a resolution $X' \rightarrow X$ of a CY complete intersection $X \subset Y$ can be solved by finding a means to evaluate the right-hand side of \cref{eqn:intY} after performing a sequence of crepant blowups $Y' \rightarrow Y$ sufficient to resolve $X'$. However, the details of the computation necessary to evaluate the right-hand side of \cref{eqn:intY} depend sensitively on the nature of the blowups comprised by the resolution. For an ordinary blowup of an algebraic variety $Y$ centered at a complete intersection of hypersurfaces in $Y$, there exist powerful techniques for computing the pushforwards of intersection products of divisors $Y'$ to the Chow ring of $Y$ \cite{fulton,Fullwood:2011zb,Fullwood:2012kj,Esole:2017kyr}, thus making it possible to express topological intersection numbers of $Y'$ in terms of the topological intersection numbers of $Y$.\footnote{See also \cite{Esole:2018tuz,Esole:2018bmf}.} By contrast, robust techniques for computing the intersection numbers of weighted blowups of algebraic varieties $Y$ do not seem to be available outside of the special case that $Y$ is a toric variety.

In this paper, we present a generalization of available techniques for computing topological intersection numbers of weighted blowups of algebraic varieties, which does not rely on the assumption that the underlying algebraic variety is toric. The central result used throughout this paper can be stated as follows: Let $f$ be a weighted blowup of an algebraic variety $Y$,
	\begin{align}
		f: Y' \rightarrow Y,
	\end{align}
 centered at a complete intersection of smooth hypersurfaces $Z = Z_1 \cap \dots \cap Z_d \subset Y $ with weights $w_1, \dots, w_d$, and let $E$ be the exceptional divisor of the blowup. We prove (see \cref{main} in \cref{sec:push}) the following formula for the pushforward of the intersection numbers of the Chow ring class $\boldsymbol E$ to the Chow ring of $Y$:
 	\begin{align}
	\label{eqn:mainformula}
\boxed{f_* \boldsymbol E^n = (-1)^{d+1} h_{n-d} \left( \frac{\boldsymbol Z_1}{w_1}, \dots, \frac{\boldsymbol Z_d}{w_d}\right)  \frac{\boldsymbol Z_1}{w_1} \cdots \frac{\boldsymbol Z_d}{w_d}}.
	\end{align}
In the above equation, $h_m(x_1,\dots,x_k)$ is the complete homogeneous symmetric polynomial of degree $m$ in the variables $x_{i=1,\dots,k}$ (with the convention that $h_m$ is identically zero for $m<0$ and $h_0 =1$), and $\boldsymbol Z_i$ are the Chow ring classes of the hypersurfaces $Z_i \subset Y$. Essentially, \cref{eqn:mainformula} converts an unknown intersection product in the Chow ring of $Y'$ to a known intersection product in the Chow ring of $Y$. We also provide a straightforward adaptation of \cref{eqn:mainformula} to arbitrary analytic functions of the class $\boldsymbol E$, which by definition can be represented as power series. These formulae can be used to determine the intersection numbers of the proper transform $X'$ of a complete intersection $X \subset Y$, and have a number of immediate applications in string theory.

In addition to proving \cref{eqn:mainformula} and performing a number of detailed cross checks, we focus on applications of \cref{eqn:mainformula} to a particular branch of string theory called `F-theory'\footnote{See \cite{Weigand:2018rez} for an excellent overview of the state of the art in F-theory compactifications.}, which formulates vacuum solutions of type IIB string theory with seven-branes and orientifolds as 12-dimensional geometric compactifications on elliptically-fibered CY varieties. Specifically, we study the topological intersection numbers of resolved elliptically-fibered CY 4-folds, $X' \rightarrow X$, in which the resolutions comprise weighted blowups.  The broad physics motivation for studying these models is to better understand the intricate interplay between the so-called `vertical' middle cohomology $H^{2,2}_{\text{vert}}(X')$ \cite{Greene:1993vm} generated by the (2,2)-forms $\text{PD}(\hat D_I) \wedge \text{PD}(\hat D_J)$ on smooth elliptically-fibered CY 4-folds and the physics of F-theory compactifications in the presence of classical magnetic gauge flux profiles, which are essential for constructing vacuum solutions whose low energy matter spectra contain chiral fermions, potentially relevant to the phenomenology of the Standard Model of particle physics. In particular, since the set of background magnetic flux profiles responsible for determining the four-dimensional chiral matter spectrum can be associated a particular sublattice of $H^{2,2}_{\text{vert}}(X')$, one can use the quadruple intersection numbers of $X'$ to determine the restriction of the intersection pairing $H^4(X') \times H^4(X') \rightarrow \mathbb Z$ to this sublatttice as a means to determine various mathematical properties of the sublattice of flux backgrounds relevant for understanding the chiral matter spectra of various F-theory vacua \cite{Jefferson:2021bid}. To this end, we use \cref{eqn:mainformula} to carry out the first known computation of the quadruple intersection numbers of weighted blowups of two specific families of singular elliptically-fibered CY 4-folds, namely the F$_4$ and Sp($6$) models\footnote{The F$_4$ and Sp($6$) models are singular Weierstrass models characterized by the presence of (resp.) $\text{IV}^{*\text{ns}}$ and $\text{I}_6^{\text{ns}}$ Kodaira singularities over a smooth divisor in the base of the elliptic fibration \cite{Katz:2011qp}. According to the physical interpretation of F-theory, these singularities indicate the presence of seven-brane stacks carrying F$_4$ and Sp($6$) gauge fields. The F$_4$ model was studied in detail in \cite{Esole:2017rgz}.}, and we comment on the significance of these results to both recent and forthcoming work on four-dimensional F-theory flux vacua \cite{JKT}.

The remainder of this paper is organized as follows: \cref{sec:notation} summarizes some of the notation conventions used throughout this paper. In \cref{sec:review}, we review salient properties of weighted blowups, beginning with a familiar definition of weighted blowups in the context of toric geometry in \cref{sec:toricweighted}, and then proceeding to a more technical discussion in \cref{sec:formal} in preparation for the proof of \cref{eqn:mainformula}. \cref{sec:push} contains the main result of this paper, \cref{main} (which proves \cref{eqn:mainformula}) and \cref{maincor}, where the latter adapts \cref{eqn:mainformula} to arbitrary analytic functions of the exceptional divisor of a weighted blowup. In \cref{sec:P2}, we carry out an easy cross check of the results of \cref{sec:push} by comparing the intersection numbers of a weighted blowup of the toric variety $\mathbb P^2$ computed with \cref{eqn:main} to an independent calculation using well-established methods in toric geometry. In \cref{sec:SU5}, we perform a ``stress test'' of the pushforward formulae in \cref{sec:push}, using \cref{eqn:mainformula} to compute the double, triple, and quadruple intersection numbers of a weighted blowup of the SU(5) model and comparing the results to various independent computations of the intersection numbers using methods from five- and three-dimensional supergravity. \cref{sec:4DFlux} discusses applications of the pushforward formulae in \cref{sec:push} to four-dimensional F-theory flux vacua, and in particular presents the first computation of the matrices of quadruple intersection numbers for resolutions of the F$_4$ and Sp($6$) models with so-called `$(4,6)$' singular elliptic fibers over codimension-two and three loci in an arbitrary base. We conclude and discuss future directions in \cref{sec:discussion}. The precise condition under which a (weighted) blowup centered at a complete intersection is given in \cref{app:crepant}.

\section{Notation}
\label{sec:notation}

In this section, we establish some of our notation conventions used throughout the paper:

\begin{itemize}
\item Given an algebraic variety $Y$, we denote its Chow ring by $A^*(Y)$. We generally denote the Chow ring class of a divisor $D \subset Y$ with bold font, i.e. $\boldsymbol D \in A^*(Y)$. We write products of Chow ring classes as $\boldsymbol D_1 \boldsymbol D_2 \cdots$. 
\item A subvariety of an algebraic variety specified by the vanishing of polynomials $p_1, \dots , p_n$ will be denoted $V(p_1,\dots, p_n)$ (i.e., `$V$' for `vanishing'.) For example, the point $[0:1]$ in $\mathbb P^1$ with homogeneous coordinates $[x:y]$ is denoted by $V(x)$.
\item We use the standard symbol $\mathbb P^d$ to denote $d$-dimensional complex projective space over an algebraically-closed field $k$ of characteristic zero, usually the complex numbers $\mathbb C$. More generally, $\mathbb P^d_{\vec w}$ denotes weighted $d$-dimensional projective space with weights $w_0,\dots, w_d \in \mathbb Z_{>0}$. We write, e.g., $[z_0:\cdots: z_d]$ for the homogeneous coordinates and represent the action of the multiplicative group $\mathbb{G}_{\text{m}}$ (usually $\mathbb C^*$) by $[z_0:\cdots :z_d] \cong [\lambda^{w_0 } z_0 :\cdots:\lambda^{w_d} z_d]$, given $\lambda \in \mathbb{G}_{\text{m}}$.  
\item For more general projective toric varieties, we continue to use, e.g., the $z_i$ for the homogeneous coordinates. In some cases, we use tables of the form 
	\begin{align}
		\begin{array}{cccc}
			 z_0 & z_1 & \cdots & z_m\\\hline
			 w_{0}^{(1)} & w_1^{(1)}& \cdots & w_m^{(1)} \\
			w_{0}^{(2)} & w_1^{(2)}& \cdots & w_m^{(2)}\\
			\vdots&\vdots&\cdots &\vdots
		\end{array}
	\end{align}
to represent the more general $\mathbb{G}_{\text{m}}^m$ quotient 
	\begin{align}
	(z_0, z_1 , \dots, z_m) \cong ((\lambda_1^{w_0^{(1)}} \lambda_2^{w_0^{(2)}}  \cdots )z_0,(\lambda_1^{w_1^{(1)}} \lambda_2^{w_1^{(2)}}\cdots) z_1,\dots, (\lambda_1^{w_m^{(1)}} \lambda_2^{w_m^{(2)}} \cdots) z_m). 
\end{align}
For example, the homogeneous coordinates $[x:y]$ of $\mathbb P^1$ can be represented as 
	\begin{align}
		\begin{array}{cc}
		x & y \\\hline
		1 & 1
		\end{array}
	\end{align}
and the homogeneous coordinates of the blowup of $\mathbb P^2$ centered at $[0:0:1]$ are given in \cref{eqn:BlP2coord}. Note that this description of the homogeneous coordinates alone does not indicate the Stanley-Reisner ideal. 
\item In the technical discussion in \cref{sec:formal}, we use $\mathbb A^d$ to dente $d$-dimensional affine space. Moreover, $T \subset \mathbb A^d$ is the $d$-dimensional algebraic torus (e.g., when $\mathbb A^d = \mathbb C^d$, then $T = (\mathbb C^*)^d$). Given affine coordinates $x_1, \dots, x_d$, we express the commutative ring whose spectrum corresponds to $\mathbb A^d$ as $k[x_1,\dots,x_d]$. 
\item The main examples in this paper consist of weighted blowups of elliptic fibrations, realized explicitly as hypersurfaces $X \subset Y$. We abuse notation and use $\hat D$ to denote both divisors of $X$, as well as their classes in $A^*(X)$. We use the Shioda-Tate-Wazir theorem to work with a preferred basis of divisors $\hat D_I$, where the index $I$ runs over various values---see \cref{sec:formal} and the discussion immediately below. 
\item We typically denote a blowup of $Y$ by $f: Y' \rightarrow Y$. Given a subvariety $X \subset Y$, we abuse notation and denote the proper transform of $X$ by $X' \subset Y'$. Given a blowup of $Y$ centered at $Z \subset Y$, we sometimes denote the blowup more explicitly by $\text{Bl}_{Z} Y$, or equivalently $\text{Bl}_{V(z)} Y$ when $Z= V(z)$; for weighted blowups, we use the notation $\text{Bl}^{\vec{w}}_{Z} Y$, etc. 
\item Later, when we discuss more involved examples of blowups, we adopt the notation 
	\begin{align}
	f_{i+1} =
  (g_{i+1,1},\dots,g_{i+1,n_{i+1}} |e_{i+1} )_{w_{i+1,1},\dots,w_{i+1,n_{i+1}}}, 
 \end{align}
 which is shorthand for the
  contraction map 
  	\begin{align}
	f_{i+1}:Y_{i+1} \rightarrow Y_i
\end{align}
 associated to the blowup of the ambient space $Y_i$
 centered at $V( g_{i+1,1}, \cdots,  g_{i+1,n_{i+1}})
\subset Y_i$ with weights $w_{i+1,1},\dots,w_{i+1,n_{i+1}}$ and exceptional divisor 
	\begin{align}
		E_{i+1} = V(e_{i+1}) \subset Y_{i+1}.
	\end{align}
Given a subvariety $X_i \subset Y_i$, the corresponding blowup $X_{i+1} \rightarrow X_i$ is then given
  by the restriction of $f_{i+1}$ to $X_{i+1}$, where $X_{i+1} \subset Y_{i+1}$ denotes the proper transform of $X_i$. We abuse notation
  and implement the $i+1$th blowup by making the substitution
  $g_{i+1,j} \rightarrow e_{i+1} g_{i+1,j}$ whenever $g_{i+1,j}$ is a homogeneous coordinate. 
  
  \item Given a pair of divisors of a CY $n$-fold, $\hat D_I, \hat D_J \subset X$, we denote the homology class (in $H_{2,2}(X)$) of their intersection by $[\hat D_I \cap \hat D_J]$. 
\end{itemize}

\section{Review of Weighted Blowups}

In this section we review weighted blowups. In \cref{sec:toricweighted}, we give a rudimentary definition of weighted blowups and use toric geometry to illustrate the distinction between ordinary and weighted blowups, before moving onto a somewhat more abstract discussion in \cref{sec:formal}.

\label{sec:review}
\subsection{Weighted blowups of toric varieties}
\label{sec:toricweighted}
In numerous examples of string compactifications, a CY $n$-fold $X$ is often concretely realized as a complete intersection of a toric variety satisfying $c_1(TX) =0$.\footnote{A famous and thoroughly studied example of such a construction is the quintic 3-fold, which is a degree five hypersurface in $Y = \mathbb P^4$ \cite{Candelas:1990rm}.} It is largely in the context of toric geometry that blowups, and more generally weighted blowups, first made their way into the string theory literature. See \cite{hori2003mirror} for a physics-oriented discussion of toric geometry, and see \cite{7d3c9f5cc0d546b3a3c103df4d4cb4f5} for more comprehensive background.

An ordinary blowup of an algebraic variety replaces a given subvariety with the projectivization of its normal bundle, leading to a quasiprojective variety. For concreteness, consider the blowup of an algebraic variety $Y$ centered at a complete intersection of smooth hypersurfaces $Z_1,\dots, Z_d$ meeting locally transversally, where $Z_i = V(z_i)$. The total transform $Y'$ is the subvariety
	\begin{align}
	\label{eqn:blowupmatrix}
		\cap_{i,j=1,\dots,d} V( z_i z'_j - z_j z'_i  )~\subset~ Y \times \mathbb P^{d-1},
	\end{align}
where the $\mathbb P^{d-1}$ factor has homogeneous coordinates 
	\begin{align}
		[ z_1': \cdots :  z_d'] \cong [ \lambda z_1' : \cdots: \lambda z_d' ]
	\end{align}
for $\lambda \in \mathbb C^*$. Equivalently, we may introduce a homogeneous coordinate $e$ whose vanishing locus is the exceptional divisor of the blowup $Y' \rightarrow Y$, and define
	\begin{align}
		 z_i  = e  z_i',
	\end{align}
where the $\mathbb C^*$ action on the new homogeneous coordinate is $e \rightarrow \lambda^{-1} e$. Now, consider the special case that $Y$ is a toric variety, and the hypersurfaces $Z_i \subset Y$ are hyperplanes given by the vanishing of affine or homogeneous coordinates. Blowups of toric varieties correspond to subdivisions of their corresponding toric fans. As a simple example, consider the blowup of $Y = \mathbb P^2$ centered at the point $[0:0:1]$, where we take $Y$ to have homogeneous coordinates $[x:y:z]$. The fan of $\mathbb P^2$ is spanned by the lattice vectors $\vec v_x = (1,0), \vec v_y = (0,1), \vec v_z = (-1,-1)$, and thus this blowup corresponds to subdividing the cone spanned by $\vec v_x, \vec v_y$ by adding a new lattice vector $\vec v_e = \vec v_x + \vec v_y$. The total transform $Y' = \text{Bl}_{[0:0:1]} \mathbb P^2$ is a Hirzebruch surface $\mathbb F_1$, which has homogeneous coordinates \cite{7d3c9f5cc0d546b3a3c103df4d4cb4f5}
	\begin{align}
	\label{eqn:BlP2coord}
		\begin{array}{cccc}   x' &  y' & z & e \\\hline  1 & 1& 1 & 0 \\  1 & 1 & 0 & -1 \end{array}
	\end{align}
and Stanley-Reisner ideal $(x'y',z e)$. The toric fan for $ Y'$ is displayed on the left in \cref{fig:BlP2fan}. 

Blowups of toric varieties restrict cleanly to blowups of their subvarieties consisting of complete intersections of hypersurfaces. For example, we can consider a nodal cubic curve $X \subset Y= \mathbb P^2$, where $X = V(y^2z- x^2 (x+ z) )$. Note that this curve is singular at the point $V(x,y)$. We can resolve this singularity by means of the blowup of $Y$ centered $V(x,y)$ that we described above. The proper transform $ \widetilde X \subset  Y' = \mathbb F_1$ is given by factoring out two copies of the exceptional divisor $V(e)$, i.e. $\widetilde X = V(  y'^2 z -  x'^2 (e  x'+  z))$, and we can see that the singular locus of $X$ has been replaced by a pair of points intersected by the exceptional divisor $\mathbb P^1 \subset  Y'$.
	
Although ordinary blowups are in many cases sufficient to resolve singularities of CY complete intersections, there are numerous examples for which weighted blowups provide the only clear method for resolving singularities in a manner that preserves the triviality of the first Chern class---see \cref{sec:4DFlux} for examples.\footnote{As mentioned in the Introduction, blowups that preserve the triviality of the first Chern class are referred to as `crepant' \cite{reid1983minimal}. The precise condition that must be satisfied for a blowup to be crepant in the class of constructions we study in this paper is given in \cref{app:crepant}.} A weighted blowup is a generalization of an ordinary blowup in which the center $Z$ is replaced by the weighted projectivization of its normal bundle. One could attempt to work by analogy and define the total transform $Y'$ of a weighted blowup of an algebraic variety $Y$ centered at a complete intersection $Z_1, \cdots, Z_d$ with weights $w_1, \dots, w_d \subset \mathbb Z_{>0}$ to be the subvariety 
	\begin{align}
		\cap_{i,j=1,\dots,d} V(z_i^{p_i} {z'_j}^{p_j} - z_j^{p_j} {z'_i}^{p_i}  )~\subset~ Y \times \mathbb P^{d-1}_{\vec w},
	\end{align}
where the factor $\mathbb P^{d-1}_{\vec w}$ is a weighted projective space with homogeneous coordinates
	\begin{align}
		[ z_1':\cdots : z_d'] \cong [\lambda^{w_1 }  z_1':\cdots : \lambda^{w_d}  z_d']
	\end{align}
and the exponents $p_i$ are defined so that $w_i p_i = w_1 \cdots w_d$. However, this definition is incomplete. A better definition, and one that has traditionally been used in the string theory literature, is as follows: introduce a homogeneous coordinate $e$ whose vanishing locus is the exceptional divisor of the blowup $Y' \rightarrow Y$, and define
	\begin{align}
		z_i  = e^{w_i}  z_i'.
	\end{align}
The $\mathbb C^*$ action on the new homogeneous coordinate is $e \to \lambda^{-1} e$. We again illustrate this operation using an example of a toric variety. Consider the weighted blowup of $Y = \mathbb P^2$ centered at $V(x,y)$, with weights $w_x=2,w_y=3$. In the toric fan, this corresponds to subdividing the cone spanned by $\vec v_x, \vec v_y$ by adding a new lattice vector with nontrivial weights, namely $\vec v_e = 2 \vec v_x + 3 \vec v_y$. The total transform $ Y'$ is a singular projective toric variety with homogeneous coordinates
	\begin{align}
	\label{eqn:homcoordBl23P2}
		\begin{array}{cccc}
			 x'&y'&z&e\\\hline
			 1 & 1 & 1 & 0 \\
			2& 3 & 0& -1
		\end{array}.
	\end{align}
The toric fan of $Y'$ is displayed on the right in \cref{fig:BlP2fan}. The intersection theory of this example is explored in further detail in \cref{sec:P2}. Again, weighted blowups of toric varieties $Y$ restrict cleanly to weighted blowups of complete intersections $X \subset Y$, with the caveat that the total transform $Y'$ is not in general smooth, due to the intrinsic quotient singularities of certain weighted projective spaces.

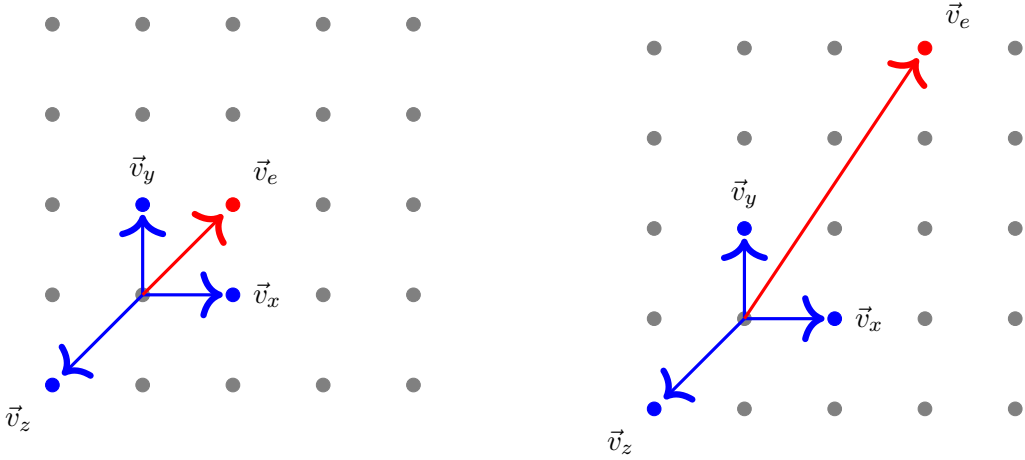
\begin{figure}
\begin{center}
$
\begin{tikzpicture}
\node[](U) at (-4,0) {$
	\begin{tikzpicture}[scale=1.2]
		\node[draw,circle,scale=.5,fill=red,color=gray] at (0,0){};
		\node[draw,circle,scale=.5,fill=red,color=blue] at (1,0){};
		\node[draw,circle,scale=.5,fill=red,color=blue] at (0,1){};
		\node[draw,circle,scale=.5,fill=red,color=gray] at (2,0){};
		\node[draw,circle,scale=.5,fill=red,color=gray] at (0,2){};
		\node[draw,circle,scale=.5,fill=red,color=gray] at (3,0){};
		\node[draw,circle,scale=.5,fill=red,color=gray] at (0,3){};
		\node[draw,circle,scale=.5,fill=yellow,color=gray] at (1,2){};
		\node[draw,circle,scale=.5,fill=red,color=gray] at (2,1){};
		\node[draw,circle,scale=.5,fill=red,color=gray] at (1,3){};
		\node[draw,circle,scale=.5,fill=red,color=gray] at (3,1){};
		\node[draw,circle,scale=.5,fill=red,color=gray] at (2,2){};
		\node[draw,circle,scale=.5,fill=red,color=gray] at (2,3){};
		\node[draw,circle,scale=.5,fill=red,color=gray] at (3,2){};
		\node[draw,circle,scale=.5,fill=red,color=gray] at (3,3){};
		\node[draw,circle,scale=.5,fill=yellow,color=red] at (1,1){};
		\node[draw,circle,scale=.5,fill=red,color=blue] at (-1,-1){};
		\node[draw,circle,scale=.5,fill=red,color=gray] at (-1,0){};
		\node[draw,circle,scale=.5,fill=red,color=gray] at (-1,1){};
		\node[draw,circle,scale=.5,fill=red,color=gray] at (-1,2){};
		\node[draw,circle,scale=.5,fill=red,color=gray] at (-1,3){};
		\node[draw,circle,scale=.5,fill=red,color=gray] at (3,-1){};
		\node[draw,circle,scale=.5,fill=red,color=gray] at (2,-1){};
		\node[draw,circle,scale=.5,fill=red,color=gray] at (1,-1){};
		\node[draw,circle,scale=.5,fill=red,color=gray] at (0,-1){};
		\node[color=blue,label={right:$\vec v_x$}] (x) at (1,0) {};
		\node[color=blue,label={above:$\vec v_y$}] (y) at (0,1) {};
		\node[color=blue,label={225:$\vec v_z$}] (z) at (-1,-1) {};
		\node[color=blue,label={45:$\vec v_e$}] (e) at (1,1) {};
		\draw[big blue arrow,very thick,color=blue] (0,0) -- (x);
		\draw[big blue arrow,very thick,color=blue] (0,0) -- (y);
		\draw[big blue arrow,very thick,color=blue] (0,0) -- (z);
		\draw[big red arrow,very thick,color=red] (0,0) -- (e);
	\end{tikzpicture}$};
\node[](W) at (4,0) {$
	\begin{tikzpicture}[scale=1.2]
		\node[draw,circle,scale=.5,fill=red,color=gray] at (0,0){};
		\node[draw,circle,scale=.5,fill=red,color=blue] at (1,0){};
		\node[draw,circle,scale=.5,fill=red,color=blue] at (0,1){};
		\node[draw,circle,scale=.5,fill=red,color=gray] at (2,0){};
		\node[draw,circle,scale=.5,fill=red,color=gray] at (0,2){};
		\node[draw,circle,scale=.5,fill=red,color=gray] at (3,0){};
		\node[draw,circle,scale=.5,fill=red,color=gray] at (0,3){};
		\node[draw,circle,scale=.5,fill=yellow,color=gray] at (1,2){};
		\node[draw,circle,scale=.5,fill=red,color=gray] at (2,1){};
		\node[draw,circle,scale=.5,fill=red,color=gray] at (1,3){};
		\node[draw,circle,scale=.5,fill=red,color=gray] at (3,1){};
		\node[draw,circle,scale=.5,fill=red,color=gray] at (2,2){};
		\node[draw,circle,scale=.5,fill=red,color=red] at (2,3){};
		\node[draw,circle,scale=.5,fill=red,color=gray] at (3,2){};
		\node[draw,circle,scale=.5,fill=red,color=gray] at (3,3){};
		\node[draw,circle,scale=.5,fill=yellow,color=gray] at (1,1){};
		\node[draw,circle,scale=.5,fill=red,color=blue] at (-1,-1){};
		\node[draw,circle,scale=.5,fill=red,color=gray] at (-1,0){};
		\node[draw,circle,scale=.5,fill=red,color=gray] at (-1,1){};
		\node[draw,circle,scale=.5,fill=red,color=gray] at (-1,2){};
		\node[draw,circle,scale=.5,fill=red,color=gray] at (-1,3){};
		\node[draw,circle,scale=.5,fill=red,color=gray] at (3,-1){};
		\node[draw,circle,scale=.5,fill=red,color=gray] at (2,-1){};
		\node[draw,circle,scale=.5,fill=red,color=gray] at (1,-1){};
		\node[draw,circle,scale=.5,fill=red,color=gray] at (0,-1){};
		\node[color=blue,label={right:$\vec v_x$}] (x) at (1,0) {};
		\node[color=blue,label={above:$\vec v_y$}] (y) at (0,1) {};
		\node[color=blue,label={225:$\vec v_z$}] (z) at (-1,-1) {};
		\node[color=blue,label={45:$\vec v_e$}] (e) at (2,3) {};
		\draw[big blue arrow,very thick,color=blue] (0,0) -- (x);
		\draw[big blue arrow,very thick,color=blue] (0,0) -- (y);
		\draw[big blue arrow,very thick,color=blue] (0,0) -- (z);
		\draw[big red arrow,very thick,color=red] (0,0) -- (e);
	\end{tikzpicture}$};
\end{tikzpicture}
$
\end{center}
\caption{\textbf{Left:} Toric fan in $\mathbb Z^2$ for the unweighted blowup of $\mathbb P^2$ centered at $x=y=0$. \textbf{Right:} Toric fan for the weighted blowup of $\mathbb P^2$ with weights $(2,3)$. We denote lattice points by pairs of coordinates $\vec v = (v_1,v_2)$, with $\vec v_x = (1,0), \vec v_y = (0,1), \vec v_z = (-1,-1)$ (in blue), and $\vec v_e = (1,1)$ for the unweighted blowup and $ \vec v_e =(2,3)$ for the weighted blowup, both in red. In the unweighted case, the exceptional divisor is a copy of $\mathbb P^1$, whereas in the weighted case, the exceptional divisor is a copy of the weighted projective line $\mathbb P^1_{(2,3)}$. Note that the weighted blowup of $\mathbb P^2$ contains cyclic quotient singularities along the exceptional divisor, and hence its intersection numbers are fractional---see, e.g., \cref{eqn:E2BlP2}.}
\label{fig:BlP2fan}
\end{figure}

\subsection{Some formal definitions}
\label{sec:formal}
In this section, we give a more formal definition of weighted blowups and related concepts, as these definitions will be useful in the proof of \cref{main}.   Readers interested primarily in applications of \cref{main} should feel free to skip the following technical discussion and apply \cref{eqn:main,eqn:maincor} to concrete examples. Note that this is the only section of the paper in which we work stack-theoretically, and that in the remainder of the paper, we simply work with the coarse moduli space of weighted blowups.

To begin, consider a complete intersection of smooth hypersurfaces $Z = Z_1 \cap \cdots \cap Z_d$ meeting locally transversally in an algebraic variety $Y$. Analytically (i.e., in an affine open set $U \subset Y$), we may write $Z_i|_{U} = V(z_i)$, where $z_1,\dots, z_d$ are affine coordinates of $U$. 

\begin{definition}[\bf{Associated Rees algebra}]
Let $I$ be the ideal sheaf associated to $Z \subset Y$, and let $I_n \subseteq I$ be the sub-ideal sheaf such that in any affine open set $U \subset Y$, $I_n|_{U}$ is generated by monomials of (weighted) degree at least $n$ in the affine coordinates $z_1,\dots, z_d$ with corresponding weights $w_1,...,w_d \subset \mathbb Z_{>0}$. (Note that we adopt the convention $I_0 := \mathcal O_Y$.) We define the \emph{Rees algebra} $R$ associated to the complete intersection $Z$ and the weights $w_1, \dots, w_d$ to be the graded algebra given by the direct sum of the ideals $I_n$, $R:= \oplus_{n\geq 0} I_n \cdot t^n$. 
\end{definition}

As a simple example, let $Y= \mathbb A^2 = k [x,y]$ (where $k$ is an algebraically-closed field of characteristic zero). Then the Rees algebra associated to $Z  = V(x,y) \subset Y$ with weights $w_x = 2, w_y = 3$ is given by the ideals
	\begin{align*}
		\begin{split}
			I_0= k[x,y], ~~~I_1 = (x,y),~~~ I_2 = (x,y),~~~ I_3 = (x^2,y),~~~I_4 = (x^2,y^2,xy), ~~~I_5 = (x^3,y^2,xy), \ \dots.
		\end{split}
	\end{align*}

\begin{definition}[\bf{Weighted blowup}] \label{def:weighted}
Given a complete intersection $Z \subset Y$ and a choice of weights $w_1, \dots, w_d \subset \mathbb Z_{>0}$, the (stack-theoretic) weighted blowup of $Y$ centered at $Z$ (with weights $w_1 , \dots, w_d$) is defined as  $$\mathscr P\text{roj}_Y(R) \ra Y$$ where $R$ is the associated Rees algebra, and $\mathscr P\text{roj}_Y$ is the stack theoretic quotient $$\text{[Spec}_Y(R) \smallsetminus V(R_+)/\G_m],$$
 where $R_+$ is the irrelevant ideal $\oplus_{n >0} I_n \cdot t^n$ and $\mathbb{G}_{\text{m}}$ is the multiplicative group characterizing the projectivization.

\end{definition}
Note that the stack-theoretic $\mathscr P\text{roj}$ is analogous to the usual scheme-theoretic $\mathscr P\text{roj}$ used to define blowups (see, e.g., \cite[Appendix B]{fulton}).

\begin{remark}
We define a weighted blowup only in the special case of the center being a complete intersection of smooth hypersurfaces meeting locally transversally in an algebraic variety. Although weighted blowups can be defined in much greater generality, \cref{def:weighted} is sufficiently general to cover the types of weighted blowups to which the results of \cref{main} apply. For additional background, including a more general definition and discussion of weighted blowups, see \cite{weighted}. 
\end{remark}

There is a special case of a weighted blowup at a divisor called a \emph{root stack}: 

\begin{definition}[\bf{Root stack}]
    Let $Z$ be a divisor in $Y$. We define the $w$th root stack of $Z$ in $Y$, $Y(\sqrt[w]{Z})$, to be the weighted blowup of $Y$ centered at $Z$ with weight $w$. 
\end{definition}

Intuitively, a root stack may be thought of as formally adjoining a root of the divisor. If one thinks of a weighted blowup as replacing the center with a weighted projective (stack) bundle, formally adjoining a root corresponds to replacing the center with a zero-dimensional weighted projective bundle.

\section{Pushforward Formula for Weighted Blowups}
\label{sec:push}

In this section, we prove the main result of this paper:

\begin{theorem} \label{main}{Let $Z \subset Y$ be the complete intersection of $d$ nonsingular hypersurfaces $Z_1, \dots , Z_d$ meeting transversally in $X$. Let $f :  Y' \rightarrow Y$ be the weighted blowup of $Y$ centered at $Z$, with weights $w_1, \dots, w_d \in \mathbb Z_{> 0}$. Let $E$ be the exceptional divisor and $\boldsymbol {E}$ its class in $A^*(Y')$. Then}
	\begin{align}
	\label{eqn:main}
		f_* \boldsymbol E^n = (-1)^{d+1} h_{n-d} \left( \frac{\boldsymbol Z_1}{w_1}, \dots, \frac{\boldsymbol Z_d}{w_d}\right)  \frac{\boldsymbol Z_1}{w_1} \cdots \frac{\boldsymbol Z_d}{w_d}
	\end{align}
{where $\boldsymbol Z_i $ denotes the class of $Z_i$ in $A^*(Y)$, $h_i(x_1,\dots,x_k)$ is the complete homogeneous symmetric polynomial of degree $i$ in $(x_1, \dots, x_k)$ with the convention that $h_i$ is identically zero for $i<0$ and $h_0 =1$.}
\end{theorem}

Our strategy for proving \cref{main} is to reduce the problem of computing the pushforward of products of exceptional divisor classes in the Chow ring of an arbitrary weighted blowup $Y'$ to the special case that $Y'$ is a weighted blowup of the stack-quotient of affine space $\mathbb A^d$ by the algebraic group 
$T = (\mathbb A^*)^d \subset \mathbb A^d$, which we denote by $[\A^d/T]$.

After proving \cref{main} in this special case, we construct a smooth map $\phi:Y \ra [\A^d/T]$ and compute the pullback (defined with respect to $\phi$) of the result to obtain the more general formula asserted in \cref{main}. Before defining the map $\phi : Y \ra [\A^d /T]$, we first provide some relevant background about the Chow ring $A^*([\A^d/T])$. 

Given a stack quotient of a scheme (in particular, affine space) by an algebraic group $G$, it has been shown in \cite[Theorem 4]{EG} that its Chow ring is isomorphic to the $G$-equivariant Chow ring as constructed in \cite[Definition 1]{EG}. Specifically, we consider the Chow ring of the quotient $[\A^d/T]$, given in terms of the $T$-equivariant Chow ring of $\A^d$:
    \begin{align*}
A^*([\mathbb A^d /  T ]) \cong A_T^*(\mathbb A^d) \cong A^*(\A^d)[\boldsymbol{t}_1,\dots, \boldsymbol{t}_d] \cong \mathbb Q [\boldsymbol{t}_1,\dots, \boldsymbol{t}_d].
    \end{align*}
 Intuitively, we may think of the $T$-equivariant geometry of $\A^d$ as the geometry of the quotient $[\A^d/T]$, and we may think of $ t_i$ as the variable keeping track of the action of the $i$th coordinate of $T$ on the $i$th coordinate $x_i$ of $\A^d$. Note that the only coordinate of the torus that acts on $x_i$ is the $i$-th, and therefore the class of $V(x_i)$ in $A^*_T(\A^d)$ is $\sum_j \delta_{ij} \boldsymbol t_j=\boldsymbol t_i$.

\begin{definition} Let $\mathscr{F}:=Isom_T(\oplus \mathcal{O}_Y(Z_i), Y \times \A^d )$ be the sheaf of $T$-equivariant vector bundle isomorphisms from $\oplus \mathcal O_Y(Z_i) \rightarrow Y \times \mathbb A^d$, i.e. $\mathscr F|_U = \{\oplus \mathcal{O}_Y(Z_i)|_U \ra U \times \A^d \}$.   
\end{definition}

Note that the points of $\mathscr F
$ are pairs $(y,t)$ with $y \in Y$ and $t \in T$ where again $T \subset \mathbb A^d $ is the $d$-dimensional torus. In particular, $\mathscr F$ is a principal $T$-bundle over $Y$.

\begin{lemma}\label{flower lemma} There exists a smooth map $\phi: Y \ra [\A^d/T]$ such that $\phi|_Z: Z \ra [0/T]$  and $\phi^*(\boldsymbol t_i)=\boldsymbol Z_i$.
\end{lemma}

\begin{proof}
Let $\sigma_i$ be the section of $\mathcal{O}_{Y}(Z_i)$ associated to $Z_i$.
We can construct a map $\Phi : \mathscr F \ra \A^d$ given by $(y,t) \mapsto (\sigma_1(y),...,\sigma_d(y))$.\footnote{Note that since we are working with stacks, a point-wise definition of the map $\phi$ is not enough, as we are working with generalized points, i.e. maps from schemes.}
This map is $T$-equivariant and we are now going to prove it is also smooth.

Let $\Phi_i: \mathscr F \ra \A^1$ defined by $(y,t) \mapsto \sigma_i(y)$, we will show that each of the $\Phi_i$ is smooth.  
Since the $Z_i$ are smooth divisors, their pre-images $Z_{i,\mathscr F}$ in $\mathscr F$ are also smooth, as are the restrictions $\Phi_i|_{Z_{i,\mathscr F}}:Z_{i,\mathscr F}\to \{ 0 \}$. Since smoothness is local, $\Phi_i$ is smooth in a neighborhood of 0, and is therefore smooth since all nonzero points are in the same orbit of the torus action. Since the $\Phi_i$ are smooth, so is $\Phi$.

By definition $\Phi$ induces a smooth map on the quotient stacks, $\phi: Y \ra [\A^d/T]$.

Note that the locus $Z_i$ is exactly the intersection of the section $\sigma_i$ with the zero section of $\mathcal{O}_Y(Z_i)$. Now, the locus $V(x_i) \subset \A^d$ pulls back exactly to the points $(Z_i,t)$ for every $t \in T$, therefore the pullback of $\boldsymbol t_i \in A^*([\A^d/T])$ is exactly $\boldsymbol{Z}_i$ as desired. 
\end{proof}

Now that we have the map $\phi$, we turn our attention to proving \cref{main} in the special case of a weighted blowup of $[\A^d/T]$ centered at $[0/T]$. The proof of this special case utilizes the fact that if one first performs a root stack of $\mathbb A^d$, and then performs an ordinary blowup of the resulting root stack, there exists a map to the total transform of the weighted blowup of $\mathbb A^d$, such that the exceptional divisor pushes forward to the exceptional divisor of the weighted blowup; see \cref{eqn:rootcommdiagram}. In order to make use of this fact, we must determine the Chow ring of the root stack of $\mathbb A^d$, as well as the action of the pushforward map defined with respect to the root stack.

\begin{lemma} \label{chow ring of the root stack}
     Let $ r:\sqrt[\vec w]{\mathbb A^d}  \ra \A^d$ be the composition of the $d$ root stacks of the divisors $V(x_i)$ with weight $w_i \in \mathbb Z_{>0}$, where we have introduced the shorthand $\sqrt[\vec w]{\mathbb A^d}$ $:=$ $\A^d(\sqrt[w_1]{V(x_1)},...,\sqrt[w_d]{V(x_d)})$. Then
    \begin{enumerate}
        \item $A^*_T(\sqrt[\vec w]{\mathbb A^d})=\mathbb Q[\boldsymbol s_1,...,\boldsymbol s_d]$
        \item $r_*: A^*_T(\sqrt[\vec w]{\mathbb A^d}) \ra A^*_T(A^d)$ maps the generators $\boldsymbol s_i \mapsto \frac{\boldsymbol t_i}{w_i}$.
    \end{enumerate}
\end{lemma}

\begin{proof}
    This follows easily from the identification $A^*_T(\sqrt[\vec w]{\mathbb A^d})\cong A^*([\sqrt[\vec w]{\mathbb A^d}/T])$. In particular, $(1)$ follows from the fact that $[\sqrt[\vec w]{\mathbb A^d}/T]$ and $[\A^d/T]$ have the same coarse moduli space, which by \cite[Proposition 6.1]{Vistoli:1989aa} implies that the Chow rings are isomorphic. Moreover $(2)$ is a straightforward application of \cite[Definition 3.6]{Vistoli:1989aa}.
 \end{proof}

\begin{lemma} \label{the other lemma}
    Let $f: \text{Bl}^{\vec{w}}_0\A^d \ra \A^d$ be the weighted blowup of $\A^d$ at the origin with weights $w_1,...,w_d$. Let $\boldsymbol E$ be the class of the exceptional divisor $E$ in the $T$-equivariant Chow ring $A_T^*(\text{Bl}^{\vec{w}}_0\A^d)$.  Then $$f_*\boldsymbol E^n=(-1)^{d+1}h_{n-d}\left(\frac {\boldsymbol t_1} {w_1},...,\frac {\boldsymbol t_d} {w_d} \right)\frac {\boldsymbol t_1} {w_1}... \frac {\boldsymbol t_d} {w_d}.$$
\end{lemma}
\begin{proof}

    As Queck and Rydh discuss in \cite{weighted}, the weighted blowup of $0$ in $\A^d$ can be thought as an ordinary blowup of the root stacks $\sqrt[\vec w]{\mathbb A^d}$ followed by a sequence of ``de-rooting''; this is to say that the weighted blowup sits into the Cartesian diagram
\begin{align}
\begin{array}{c}
\label{eqn:rootcommdiagram}
	\begin{tikzpicture}
		\node[](1) at (0,0) {$\text{Bl}_{\sqrt[\vec{w}]{0}}(\sqrt[\vec{w}]{\mathbb A^d})$};
		\node[](2) at (3,0) {$\text{Bl}_0^{\vec w} \mathbb A^d $};
		\node[](3) at (0,-2) {$\sqrt[\vec w]{\mathbb A^d}$};
		\node[](4) at (3,-2) {$\A^d$};
		\draw[big arrow] (1) -- node[above,midway]{\small $r'$} (2);
		\draw[big arrow] (2) -- node[right,midway]{\small $f$} (4);
		\draw[big arrow] (1) -- node[left,midway]{\small $g$} (3);
		\draw[big arrow] (3) -- node[above,midway]{\small $r$} (4);
	\end{tikzpicture}
\end{array}
\end{align}
where $r,r'$ are sequences of root stacks of the divisors (respectively) $V(x_i), \widetilde{V(x_i)}$ with weights $w_i$, $g$ is the ordinary blowup of $\sqrt[\vec w]{\mathbb A^d}$ centered at $\sqrt[\vec{w}]{0}=\sqrt[w_1]{V(x_1)} \cap \cdots \cap \sqrt[w_d]{V(x_d)}$, and $f$ is the weighted blowup of $\A^d$ at the origin with weights $w_1,...,w_d$. 

Let $F$ be the exceptional divisor of $g$. 
Since in the sequence of root stacks in $r'$ we do not take the root stack of the exceptional divisor $E$, it follows that $r'_*(\boldsymbol F^n)=\boldsymbol E^n$ and $f_*(r'_*(\boldsymbol F^n))=f_*(\boldsymbol E^n)$. 
On the other hand, by \cite[Lemma 3.4]{Esole:2017kyr} $r_*(g_* (\boldsymbol F^n))=r_*((-1)^{d+1}h_{n-d}(\boldsymbol s_1,\dots ,\boldsymbol s_d)\boldsymbol s_1 \cdots \boldsymbol s_d)$.\footnote{Technically, the proof of \cite[Lemma 3.4]{Esole:2017kyr} is given for schemes. Fortunately, the same proof works for stacks, as well.} 
By commutativity of pushforward and composition, and by \cref{chow ring of the root stack}, we obtain the desired formula. 
\end{proof}

(See \cref{fig:stackyfan} for a toric example of the Cartesian diagram in \cref{eqn:rootcommdiagram}.) We are finally ready to prove \cref{main}:
\begin{proof}
Blowups commute with smooth base change, hence we have the following Cartesian diagram
\begin{align*}
\begin{array}{c}
	\begin{tikzpicture}
		\node[](1) at (0,0) {$Y'$};
		\node[](2) at (3,0) {$ \left[\text{Bl}^{\vec{w}}_0\A^d/T \right]$};
		\node[](3) at (0,-2) {$Y$};
		\node[](4) at (3,-2) {$ \left[\A^d/T\right]$};
		\draw[big arrow] (1) -- node[above,midway]{\small $$} (2);
		\draw[big arrow] (2) -- node[right,midway]{\small $$} (4);
		\draw[big arrow] (1) -- node[left,midway]{\small $$} (3);
		\draw[big arrow] (3) -- node[above,midway]{\small $\phi$} (4);
	\end{tikzpicture}
\end{array}
\end{align*}
where the exceptional divisor on the right pulls back to the exceptional divisor $E \in Y'$.
By applying \cref{the other lemma} and pulling back via $\phi^*$ we obtain the desired equality.
\end{proof}
\begin{figure}
\begin{center}
$
\begin{tikzpicture}[scale=1.2]
\node[](U) at (-3,6) {$
\begin{tikzpicture}[scale=1.1]
		\node[draw,circle,scale=.5,fill=red,color=gray] at (0,0){};
		\node[draw,circle,scale=.5,fill=red,color=gray] at (1,0){};
		\node[draw,circle,scale=.5,fill=red,color=gray] at (0,1){};
		\node[draw,circle,scale=.5,fill=red,color=blue] at (2,0){};
		\node[draw,circle,scale=.5,fill=red,color=gray] at (0,2){};
		\node[draw,circle,scale=.5,fill=red,color=gray] at (3,0){};
		\node[draw,circle,scale=.5,fill=red,color=blue] at (0,3){};
		\node[draw,circle,scale=.5,fill=yellow,color=gray] at (1,2){};
		\node[draw,circle,scale=.5,fill=red,color=gray] at (2,1){};
		\node[draw,circle,scale=.5,fill=red,color=gray] at (1,3){};
		\node[draw,circle,scale=.5,fill=red,color=gray] at (3,1){};
		\node[draw,circle,scale=.5,fill=red,color=gray] at (2,2){};
		\node[draw,circle,scale=.5,fill=red,color=red] at (2,3){};
		\node[draw,circle,scale=.5,fill=red,color=gray] at (3,2){};
		\node[draw,circle,scale=.5,fill=red,color=gray] at (3,3){};
		\node[draw,circle,scale=.5,fill=yellow,color=gray] at (1,1){};
		\node[color=blue,label={right:$\vec v_{\sqrt{x}}$}] (x) at (2,0) {};
		\node[color=blue,label={above:$\vec v_{\sqrt[3]{y}}$}] (y) at (0,3) {};
		\node[color=blue,label={45:$\vec v_e$}] (e) at (2,3) {};
		\draw[big blue arrow,very thick,color=blue] (0,0) -- (x);
		\draw[big blue arrow,very thick,color=blue] (0,0) -- (y);
		\draw[big red arrow,very thick,color=red] (0,0) -- (e);
	\end{tikzpicture}$};
\node[](W) at (3,6) {$
	\begin{tikzpicture}[scale=1.1]
		\node[draw,circle,scale=.5,fill=red,color=gray] at (0,0){};
		\node[draw,circle,scale=.5,fill=red,color=blue] at (1,0){};
		\node[draw,circle,scale=.5,fill=red,color=blue] at (0,1){};
		\node[draw,circle,scale=.5,fill=red,color=gray] at (2,0){};
		\node[draw,circle,scale=.5,fill=red,color=gray] at (0,2){};
		\node[draw,circle,scale=.5,fill=red,color=gray] at (3,0){};
		\node[draw,circle,scale=.5,fill=red,color=gray] at (0,3){};
		\node[draw,circle,scale=.5,fill=yellow,color=gray] at (1,2){};
		\node[draw,circle,scale=.5,fill=red,color=gray] at (2,1){};
		\node[draw,circle,scale=.5,fill=red,color=gray] at (1,3){};
		\node[draw,circle,scale=.5,fill=red,color=gray] at (3,1){};
		\node[draw,circle,scale=.5,fill=red,color=gray] at (2,2){};
		\node[draw,circle,scale=.5,fill=red,color=red] at (2,3){};
		\node[draw,circle,scale=.5,fill=red,color=gray] at (3,2){};
		\node[draw,circle,scale=.5,fill=red,color=gray] at (3,3){};
		\node[draw,circle,scale=.5,fill=yellow,color=gray] at (1,1){};
		\node[color=blue,label={right:$\vec v_{x}$}] (x) at (1,0) {};
		\node[color=blue,label={above:$\vec v_{y}$}] (y) at (0,1) {};
		\node[color=blue,label={45:$\vec v_e$}] (e) at (2,3) {};
		\draw[big blue arrow,very thick,color=blue] (0,0) -- (x);
		\draw[big blue arrow,very thick,color=blue] (0,0) -- (y);
		\draw[big red arrow,very thick,color=red] (0,0) -- (e);
	\end{tikzpicture}$};
	\node[](BD) at (-3,0) {$
	\begin{tikzpicture}[scale=1.1]
		\node[draw,circle,scale=.5,fill=red,color=gray] at (0,0){};
		\node[draw,circle,scale=.5,fill=red,color=gray] at (1,0){};
		\node[draw,circle,scale=.5,fill=red,color=gray] at (0,1){};
		\node[draw,circle,scale=.5,fill=red,color=blue] at (2,0){};
		\node[draw,circle,scale=.5,fill=red,color=gray] at (0,2){};
		\node[draw,circle,scale=.5,fill=red,color=gray] at (3,0){};
		\node[draw,circle,scale=.5,fill=red,color=blue] at (0,3){};
		\node[draw,circle,scale=.5,fill=yellow,color=gray] at (1,2){};
		\node[draw,circle,scale=.5,fill=red,color=gray] at (2,1){};
		\node[draw,circle,scale=.5,fill=red,color=gray] at (1,3){};
		\node[draw,circle,scale=.5,fill=red,color=gray] at (3,1){};
		\node[draw,circle,scale=.5,fill=red,color=gray] at (2,2){};
		\node[draw,circle,scale=.5,fill=red,color=gray] at (2,3){};
		\node[draw,circle,scale=.5,fill=red,color=gray] at (3,2){};
		\node[draw,circle,scale=.5,fill=red,color=gray] at (3,3){};
		\node[draw,circle,scale=.5,fill=yellow,color=gray] at (1,1){};
		\node[color=blue,label={right:$\vec v_{\sqrt{x}}$}] (x) at (2,0) {};
		\node[color=blue,label={above:$\vec v_{\sqrt[3]{y}}$}] (y) at (0,3) {};
		\draw[big blue arrow,very thick,color=blue] (0,0) -- (x);
		\draw[big blue arrow,very thick,color=blue] (0,0) -- (y);
	\end{tikzpicture}$};
	\node[](rBD) at (3,5.76-6) {$
	\begin{tikzpicture}[scale=1.1]
		\node[draw,circle,scale=.5,fill=red,color=gray] at (0,0){};
		\node[draw,circle,scale=.5,fill=red,color=blue] at (1,0){};
		\node[draw,circle,scale=.5,fill=red,color=blue] at (0,1){};
		\node[draw,circle,scale=.5,fill=red,color=gray] at (2,0){};
		\node[draw,circle,scale=.5,fill=red,color=gray] at (0,2){};
		\node[draw,circle,scale=.5,fill=red,color=gray] at (3,0){};
		\node[draw,circle,scale=.5,fill=red,color=gray] at (0,3){};
		\node[draw,circle,scale=.5,fill=yellow,color=gray] at (1,2){};
		\node[draw,circle,scale=.5,fill=red,color=gray] at (2,1){};
		\node[draw,circle,scale=.5,fill=red,color=gray] at (1,3){};
		\node[draw,circle,scale=.5,fill=red,color=gray] at (3,1){};
		\node[draw,circle,scale=.5,fill=red,color=gray] at (2,2){};
		\node[draw,circle,scale=.5,fill=red,color=gray] at (2,3){};
		\node[draw,circle,scale=.5,fill=red,color=gray] at (3,2){};
		\node[draw,circle,scale=.5,fill=red,color=gray] at (3,3){};
		\node[draw,circle,scale=.5,fill=yellow,color=gray] at (1,1){};
		\node[color=blue,label={right:$\vec v_{x}$}] (x) at (1,0) {};
		\node[color=blue,label={above:$\vec v_{y}$}] (y) at (0,1) {};
		\draw[big blue arrow,very thick,color=blue] (0,0) -- (x);
		\draw[big blue arrow,very thick,color=blue] (0,0) -- (y);
	\end{tikzpicture}$};
	\draw[big arrow] (BD) -- node[above,midway,scale=.8]{$r$}++ (4.2,0); 
	\draw[big arrow] (U) --node[above,midway,scale=.8]{$r'$} ++ (4.2,0); 
	\draw[big arrow] (W) --node[right,pos=.58,scale=.8]{$f$} (rBD); 
	\draw[big arrow] (U) --node[left,midway,scale=.8]{$g$} (BD); 
\end{tikzpicture}
$
\end{center}
\caption{The proof of \cref{main} exploits the fact that a weighted blowup of a toric space can be viewed as composition of a root stack, an ordinary blowup, and a ``de-rooting'', as depicted in \cref{eqn:rootcommdiagram}. We illustrate this composition of maps using (stacky) toric fans in the lattice $\mathbb Z^2$ with coordinates $\vec v =(v_1,v_2)$, as illustrated on p. 43 of \cite{weighted}. \textbf{Lower right:} Ordinary fan for $\mathbb C^2$. \textbf{Lower left:} Stacky fan for $\sqrt[\vec w]{\mathbb C^2}$, where $w_x = 2, w_y = 3$. \textbf{Upper left:} Ordinary blowup of $\sqrt[\vec w]{\mathbb C^2}$, centered at the origin $V(\sqrt{x},\sqrt[3]{y})$. \textbf{Upper right:} Weighted blowup of $\mathbb C^2$ centered at the origin. Notice that the weighted blowup, which can be implemented by the substitution $x \to e^2 x, y \to e^3 y$ is formally equivalent to the substitution $\sqrt{x} \to e \sqrt{x} , \sqrt[3]{y} \to e \sqrt[3]{y}$, which is consistent with the root stack $r'$.}
\label{fig:stackyfan}
\end{figure}
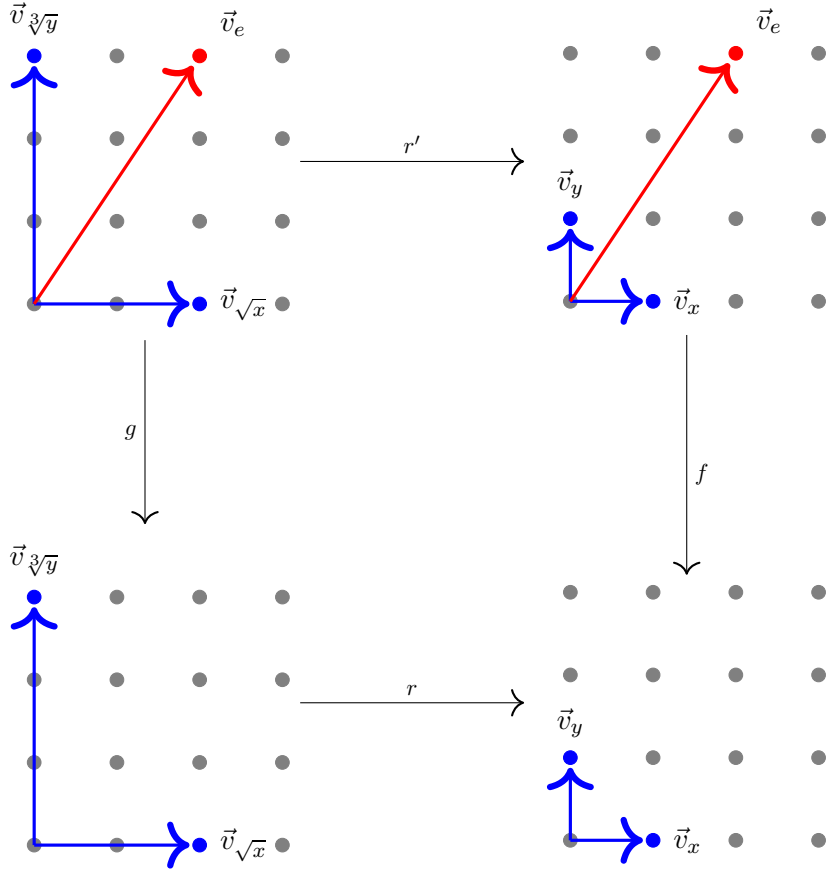

\begin{corollary} \label{maincor} Let $Z \subset Y$ be the complete intersection of $d$ smooth hypersurfaces $Z_1, \dots , Z_d$ meeting locally transversally in $Y$. Let $f :  Y' \rightarrow Y$ be the weighted blowup of $Y$ centered at $Z$, with weights $w_1, \dots, w_d \in \mathbb Z_{>0}$. We denote the class of the exceptional divisor of $f$ by $\boldsymbol E$, and the classes of the hypersurfaces by $\boldsymbol Z_1, \dots, \boldsymbol Z_d$. Let $  Q'(t) =  \sum_a f^*  Q_a t^a$ be a formal power series with $ Q_a \in A^*(Y)$. We define the associated formal power series $ Q(t) =\sum_a  Q_a t^a$ whose coefficients pull back to the coefficients of $ Q'(t)$. Then the pushforward $f_*  Q'(\boldsymbol E)$ is
	\begin{align}
	\label{eqn:maincor}
		f_*  Q'(\boldsymbol E) = \sum_{n=1}^d Q\left(\frac{\boldsymbol Z_n}{w_n}\right)  M_n, ~~~~ M_n = \prod_{\substack{m=1\\m\ne n}}^d \frac{\frac{\boldsymbol Z_m}{w_m}}{\frac{\boldsymbol Z_m}{w_m} - \frac{\boldsymbol Z_n}{w_n}}.
	\end{align}	
\end{corollary}
\begin{proof}
This is completely analogous to the proof of \cite[Theorem 1.8]{Esole:2017kyr}.
\end{proof}

\begin{remark}[\bf{``Fractional'' weights}] \label{rem:frac}
As discussed in the Introduction, the motivating physics application of \cref{eqn:main,eqn:maincor} is computing pushforwards of intersection numbers of weighted crepant blowups of complete intersection CY subvarieties $X= X_1 \cap \cdots \cap X_p \subset Y$, where $Y$ is an algebraic variety, often projective toric. It is sometimes the case that a weighted blowup centered at a complete intersection of hyperplanes $Z = Z_1 \cap \cdots \cap Z_d$ must necessarily be performed with ``fractional weights'' 
	\begin{align}
		w_i \in \mathbb Q_{>0}
	\end{align}
in order to satisfy \cref{eqn:crepantcond} and thus preserve the CY condition. In practice, such fractional weighted blowups $Y' \rightarrow Y$ lead to complete intersections $X' \subset Y''$ where the proper transforms $X_i' \subset Y''$ are the vanishing loci of well-defined sections (i.e., sections of integer tensor powers of well-defined line bundles) and $Y'' \cong Y'$ is related by elementary isomorphisms such as $\mathbb P_{1,w,\dots,w}^d \cong \mathbb P^d,\mathbb P^d_{w,w, \dots, w} \cong \mathbb P^d$, with $w\in \mathbb Z_{>0}$---see, e.g., \cite[Section 5.5]{Esole:2014hya} (this example is explored extensively in \cref{sec:SU5}). It is tempting to regard fractional weighted blowups satisfying these conditions as a formal trick that leads to a well-defined subvariety $X' \subset Y''$, where $Y''$ and $Y'$ are related by a sequence of root stacks, hence $A^*(Y'') = A^*(Y')$, and moreover $Y''$ can be regarded as an (integer) weighted blowup of an algebraic stack related similarly to $Y$ that restricts to a crepant weighted blowup of the image of $X$. In the examples discussed in this paper, we simply observe that this trick appears to lead to sensible spaces, and we leave a more careful and systematic treatment of fractional weights to future work.

\end{remark}

\section{Cross Check: Weighted Blowup of $\mathbb P^2$}

\label{sec:P2}

Here, we test the pushforward formulae in \cref{eqn:main,eqn:maincor} on a simple example, namely a weighted blowup of $\mathbb P^2$, for which the intersection numbers can be computed independently using well-known methods in toric geometry. This will serve as our first cross check of these generalized pushforward formulae.

Let $Y = \mathbb P^2$ with homogeneous coordinates $[x:y:z]$, and let
$\boldsymbol{H}	= c_1(\mathcal O_{\mathbb P^2}(1)) \in A^*(\mathbb P^2)$ denote the class of $V(x)$, satisfying
	\begin{align}
		\boldsymbol{H}^2 = 1. 
	\end{align}
We blow up $\mathbb P^2$ with weights $w_x = 2, w_y = 3$ at the point $[0:0:1]$:
	\begin{align}
		f : Y'  \rightarrow Y,~~~~ Y' = \text{Bl}_{[0:0:1]}^{\vec w} \mathbb P^2,
	\end{align}
In practice, it is convenient to implement the above blowup by abusing notation and making the replacements
	\begin{align}	
		x \rightarrow e^2 x, ~~~~ y \rightarrow e^3 y. 
	\end{align}
The variety $Y'$ has homogeneous coordinates given in \cref{eqn:homcoordBl23P2}.
We denote by $\boldsymbol{E}$ the class of the exceptional divisor $E = V(e)$.

Our task is to compute the intersection numbers of $Y'$ using \cref{eqn:main}. In this example, the only nontrivial intersection number is $\boldsymbol{E}^2$. Let $\boldsymbol{x}, \boldsymbol{y} \in A^*(\mathbb P^2)$ denote (resp.) the classes of $V(x),V(y)$. Then, the pushforward is given by 
	\begin{align}
		f_{*} \boldsymbol E^2 = \left( \frac{\boldsymbol{x}}{2} \right)^2  \frac{\frac{\boldsymbol{y}}{3}}{\frac{\boldsymbol{y}}{3} - \frac{\boldsymbol{x}}{2}} +  \left( \frac{\boldsymbol{y}}{3} \right)^2  \frac{\frac{\boldsymbol{x}}{2} }{\frac{\boldsymbol{x}}{2} - \frac{\boldsymbol{y}}{3} } = - \frac{1}{6} \boldsymbol{x}  \boldsymbol{y}.
	\end{align}
Using the fact that $\boldsymbol{x} = \boldsymbol{y} = \boldsymbol{H}$, the above equation implies
	\begin{align}
	\label{eqn:E2BlP2}
		\boldsymbol{E}^2 =- \frac{1}{6} 
	\end{align}
We check the above result by comparing our answer to the answer we obtain using known methods for computing intersection numbers of weighted blowups of projective toric varieties. First, we construct the toric fan for $Y'$ using the toric description in \cref{eqn:homcoordBl23P2}---see \cref{fig:BlP2fan}. Denoting the toric divisors of $Y'$ by $\boldsymbol{D}_{x}, \boldsymbol D_y, \boldsymbol D_z, \boldsymbol D_e$, we find the following standard linear relations $\sum_I \boldsymbol{D}_I \vec v_{I} =0$ among the divisor classes of $Y'$:
	\begin{align}
	\begin{split}
		0&=\boldsymbol{D}_x -\boldsymbol{D}_z + 2 \boldsymbol{D}_e\\
		0 &= \boldsymbol{D}_y - \boldsymbol{D}_z + 3 \boldsymbol{D}_e.
	\end{split}
	\end{align}
Combining the above linear relations with the intersection numbers
	\begin{align}
	\begin{split}
		\boldsymbol{D}_x  \boldsymbol{D}_e &= \text{det}\begin{pmatrix} \vec v_x & \vec v_e\end{pmatrix}^{-1} = \frac{1}{3}\\
		\boldsymbol{D}_e  \boldsymbol{D}_y &= \text{det}\begin{pmatrix} \vec v_e & \vec v_y\end{pmatrix}^{-1} = \frac{1}{2}\\
		\boldsymbol{D}_y  \boldsymbol{D}_z &= \text{det}\begin{pmatrix} \vec v_y & \vec v_z\end{pmatrix}^{-1} = 1 \\
		\boldsymbol{D}_z \boldsymbol{D}_x &= \text{det}\begin{pmatrix} \vec v_z & \vec v_x\end{pmatrix}^{-1} = 1,
	\end{split}
	\end{align}
and the Stanley-Reisner ideal
	\begin{align}
	\begin{split}
		\boldsymbol D_z  \boldsymbol D_e= \boldsymbol D_y  \boldsymbol D_x=0,
	\end{split}
	\end{align}
	we infer that
	\begin{align}
		\boldsymbol{D}_e^2 =\frac{1}{2} \boldsymbol{D}_e  ( \boldsymbol D_z - \boldsymbol D_x) = - \frac{1}{6} = \boldsymbol{E}^2
	\end{align}
as expected. The full identification between divisor classes in the toric and algebraic constructions is:
	\begin{align}
		\boldsymbol{D}_z \leftrightarrow \boldsymbol{H}, ~~~~ \boldsymbol{D}_e \leftrightarrow \boldsymbol{E}. 
	\end{align}

\section{Stress Test: SU(5) model}
\label{sec:SU5}

We next turn our attention to a much more non-trivial example, namely a resolution of the singular SU(5) model (defined below). The SU(5) model is our primary ``stress test'' of the new pushforward formulae presented in \cref{sec:push}, as this model describes a broad class of elliptically-fibered CY varieties whose geometric and physical properties have been studied thoroughly in the context of both global F-theory compactifications (i.e., strings propagating on spacetimes with compact internal dimensions, corresponding to supergravity theories at low energies) and local F-theory compactifications (i.e., strings propagating on spacetimes with non-compact internal dimensions, corresponding to non-gravitational supersymmetric quantum field theories at low energies) of various dimensions--- see, e.g., \cite{Beasley:2008dc,Donagi:2008kj,Esole:2011sm,Grimm:2011fx,Hayashi:2013lra}. Specifically, for both CY 3-fold and 4-fold compactifications, a subset of the topological intersection numbers of all resolutions of the SU(5) model have been computed independently using properties of supersymmetric Yang-Mills theories with SU(5) gauge symmetry. Note that although the full set of resolutions\footnote{Note that these resolutions resolve singular fibers through codimension-two components of the discriminant locus, for generic choice of complex structure (i.e. generic choices of parameters in the defining Weierstrass model.} of the SU(5) model were determined in \cite{Esole:2014hya}, the intersection numbers have only been computed geometrically for a subset of these resolutions \cite{Esole:2015xfa,Jefferson:2021bid}. To the authors' knowledge, the results we present in this section are the first general computation of pushforwards of intersection numbers for a resolution of the SU(5) model comprising weighted blowups.

\subsection{Singular model and resolution}

The SU(5) model is a special example of a singular Weierstrass model. Consider a singular elliptic fibration over a base $B$, in which the fiber of each point in $B$ is defined by the vanishing locus of a cubic in $\mathbb P^2$,
\begin{align}
\label{eqn:Tateform}
		V(y^2 z + a_1 xy z + a_3 y z^2 - (x^3 + a_2 x^2 z+ a_4 x z^2 + a_6 z^3 )).
	\end{align}
Above, $[x:y:z]$ are homogeneous coordinates of $\mathbb P^2$, and the coefficients $a_n$ can be regarded (locally) as functions on $B$, which we assume to be a smooth projective variety of arbitrary dimension. The elliptic fiber becomes singular along the vanishing locus of the discriminant of the cubic in \cref{eqn:Tateform},
	\begin{align}
		V(\Delta) \subset B,
	\end{align}
where
	\begin{align}
	\begin{split}
			f &= -\frac{1}{3} \left(\frac{a_1^2}{4}+a_2\right){}^2+\frac{a_1 a_3}{2}+a_4\\
			g&= \frac{2}{27} \left(\frac{a_1^2}{4}+a_2\right){}^3-\frac{1}{3} \left(\frac{a_1 a_3}{2}+a_4\right) \left(\frac{a_1^2}{4}+a_2\right)+\frac{a_3^2}{4}+a_6\\
			\Delta &=4 f^3 + 27 g^2.
	\end{split}	
	\end{align}
Given a choice of line bundle $\mathcal O(L) \rightarrow B$, we can regard this elliptic fibration as a hypersurface
	\begin{align}
		X_0 \subset Y_0,~~~~
		Y_0 = \mathbb P_B( \mathcal O(2L) \oplus \mathcal O(3L) \oplus \mathcal O ), ~~~~
	\end{align} 
where $Y_0$ is equipped with the projection map 
	\begin{align}
		\varpi : Y_0 \rightarrow B.
	\end{align}
In this construction, the homogeneous coordinates $x,y,z$ and the coefficients $a_n$ are sections of the following line bundles:
	\begin{align}
	\begin{split}
	z &\in \Gamma(\mathcal O(1))\\
	x &\in \Gamma(\mathcal O(1) \otimes \varpi^* \mathcal  O(2L) )\\
	y &\in \Gamma (\mathcal O(1) \otimes \varpi^* \mathcal O(3L))\\
	a_n &\in \Gamma (\varpi^* \mathcal O(nL)).
	\end{split}
	\end{align}	
For $X_0$ to define a vacuum solution of F-theory, we must impose sufficient conditions for $X_0$ to be a CY variety, namely we must require that the first Chern class of $X_0$ vanishes. It is a straightforward calculation to show that the CY condition can be imposed by setting the defining line bundle of the Weierstrass model to be equal to the anticanonical bundle of $B$, 
	\begin{align}
		L = -K_B \implies c_1(X_0) = 0. 
	\end{align}

The SU(5) model is a special case of the above construction that describes a family of singular elliptic fibrations $X_0$. This model is realized explicitly by imposing the following conditions on the sections $a_n$:
	\begin{align}
	\label{eqn:SU5tuning}
		a_1 = a_{1,0},~~~~ a_2 = a_{2,1} s,~~~~ a_3 = a_{3,3} s^2,~~~~ a_4 = a_{4,4} s^3,~~~~ a_6 = a_{6,5} s^5,
	\end{align}
where 
	\begin{align}
		S=V( s ) \subset B
	\end{align}
is a smooth divisor. (Note that by assumption $a_{m,n}$ does not vanish at a generic point of $S$.) The conditions in \cref{eqn:SU5tuning} introduce a Kodaira singularity of type I$_5^{\text{split}}$\footnote{According to a celebrated result due to Kodaira and N\'eron, the possible singularity types of the elliptic fibers over codimension-one components $S$ of the discriminant locus $V(\Delta)$ are classified (roughly) by simple Lie algebras \cite{KodairaII,KodairaIII,NeronClassification}. In F-theory, singular elliptic fibers signal the presence of seven-branes wrapping the discriminant locus $V(\Delta)$ and the singularity types of the elliptic fibers determine various kinematic properties of the seven-branes---for example, the singularity types of a fiber over a generic point in $S$ indicates the type of gauge bundle that ``lives'' on the worldvolume of the configuration of seven-brane wrapping $S$. \label{Kodaira}} in the elliptic fibers over a generic point in $S$, as can be seen by studying the sections $f,g$ and the discriminant $\Delta$ \cite{Bershadsky:1996nh,Katz:2011qp}\footnote{See also \cite{Huang:2018gpl} for a more recent update.}:
	\begin{align}
	 \Delta = s^5 (\cdots).
	\end{align}

We consider a particular resolution\footnote{This may or may not be a partial resolution, depending on the dimension of $B$ and the genericity of the choice of sections $a_n$.} of the SU(5) model, which comprises the following sequence of weighted blowups of the ambient projective bundle $Y_0$:
	\begin{align}
	\label{eqn:SU5resolution}
		f_1 \circ f_2 \circ f_3 \circ f_4 : Y_4 \rightarrow Y_0,
	\end{align}
where
	\begin{align}
		\begin{split}
		\label{eqn:SU5resolution}
			 f_1 = (x,y,s|e_1)_{2,2,1},~~f_2 = (y,e_1|e_2)_{1,1},~~ f_3=(y,s,e_2|e_3)_{\frac{1}{2},\frac{1}{2},\frac{1}{2}} ,~~ f_4 =  (s,e_1|e_4)_{\frac{1}{2},\frac{1}{2}}.
		\end{split}
	\end{align}
The geometry of the proper transform of $X_0$, namely $X_4 \subset Y_4$, is described in detail in Section 5.5 of \cite{Esole:2014hya}. 

\subsection{Intersection theory}
\label{sec:SU5int}
Our goal in this subsection is to compute the intersection numbers of divisors of the resolved elliptic fibration $X_4$. We evaluate these intersection products explicitly by computing their pushforward to the Chow ring of $B$, which we assume to be known. 

Let us begin by choosing a convenient basis of divisors. According to the Shioda-Tate-Wazir theorem \cite{wazir2004arithmetic}, the resolved elliptic fibration $X_4$ has a basis of divisors 
	\begin{align}
	\label{eqn:Xbasis}
		\hat D_I = \hat D_0, \hat D_\alpha, \hat D_i
	\end{align}
where $\hat D_0$ is the class of the holomorphic zero section $V(x , z )$, $\hat D_\alpha$ are the pullbacks of divisors in $B$, and $\hat D_i$ are called `Cartan divisors'. The ambient projective bundle $Y_4$ likewise has a complete basis of divisors, whose classes in $A^*(Y_4)$ we denote by
	\begin{align}
		 \boldsymbol{H}, \boldsymbol{D}_\alpha, \boldsymbol{E}_i,
	\end{align}
where $\boldsymbol{H} = c_1(\mathcal O_{Y_0}(1))$, $\boldsymbol{D}_\alpha$ is the class of the pullback of a divisor $D_\alpha 
\subset B$, and $\boldsymbol{E}_i$ is the exceptional divisor of the $i$th blowup. Since $X_4$ is a hypersurface of $Y_4$, it is possible to express the divisors $\hat D_I \subset X_4$ as complete intersections in $Y_4$: 
	\begin{align}
	\begin{split}
		\hat D_0 ={D}_0\cap  {X}_4,~~~~\hat D_\alpha &= {D}_\alpha \cap {X}_4,~~~~ \hat D_i = {D}_i \cap {X}_4,
	\end{split}
	\end{align}
where we have introduced the definitions
	\begin{align}
	\begin{split}	
	\label{eqn:SU5Dbasis}	
		\boldsymbol{D}_0 &= \frac{1}{3} \boldsymbol{H} \\
		\boldsymbol{D}_1 &=  \boldsymbol{E}_4 \\
			\boldsymbol{D}_2 &=\boldsymbol{E}_1 - \boldsymbol{E}_2 -\frac{1}{2} \boldsymbol{E}_4 \\
			\boldsymbol{D}_3&= \boldsymbol{E}_2 - \frac{1}{2} \boldsymbol{E}_3    \\
		\boldsymbol{D}_4& = \boldsymbol{E}_3,
	\end{split}
	\end{align}
and
	\begin{align}
	\label{eqn:X4class}
		\boldsymbol{X}_4 &= 3 \boldsymbol{H} + 6 \boldsymbol{L} - 4 \boldsymbol{E}_1 - \boldsymbol{E}_2 - \frac{1}{2} \boldsymbol{E}_3. 
	\end{align}

Now that we have a complete basis of divisors $\hat D_I \subset X_4$, along with an expression for each divisor $\hat D_I$ as a complete intersection in $Y_4$, we can state our aim more precisely. Let the dimension of the base $B$ be $d-1$, and consider the intersection numbers
	\begin{align}
	\label{eqn:X4intersection}
		\hat D_{I_1} \cdots \hat D_{I_d},
		\end{align}
which lift to the following intersection products in the Chow ring of $Y_4$:
	\begin{align} 
		\label{eqn:Y4intersection}
		\boldsymbol{D}_{I_1} \cdots \boldsymbol{D}_{I_d}  \boldsymbol{X}_4.
	\end{align}
Our aim is to use the results of \cref{sec:push} to push the above intersection products down to the Chow ring of $B$, where the pushforward is defined with respect to the projection
	\begin{align}
		\varpi \circ f_1 \circ f_2 \circ f_3 \circ f_4 : Y_4 \rightarrow B.
	\end{align}
	 The pushforward map can be expressed, analogously, as a composition
	\begin{align}
	\label{eqn:SU5pushmap}
		\varpi_* \circ f_{1*} \circ \cdots \circ f_{4*}. 
	\end{align}
We apply \cref{eqn:maincor} using the following Chow ring classes of the hypersurfaces at which the blowups $f_1,f_2,f_3,f_4$ are centered:
	\begin{align}
	\begin{split}
	\label{eqn:generators}
		\frac{\boldsymbol{Z}_{1,1}}{w_{1,1}} &= \frac{1}{2} (\boldsymbol{H} + 2 \boldsymbol{L}),~~~~ \frac{\boldsymbol{Z}_{1,2}}{w_{1,2}} = \frac{1}{2} (\boldsymbol{H} + 2 \boldsymbol{L}),~~~~ \frac{\boldsymbol{Z}_{1,3}}{w_{1,3}}= \boldsymbol{S}\\
		\frac{\boldsymbol{Z}_{2,1}}{w_{2,1}} &= \boldsymbol{H}+ 3 \boldsymbol{L} - 2 \boldsymbol{E}_1,~~~~ \frac{\boldsymbol{Z}_{2,2}}{w_{2,2}} = \boldsymbol{E}_1 \\
		\frac{\boldsymbol{Z}_{3,1}}{w_{3,1}} &= 2 \left(\boldsymbol{H}+3 \boldsymbol{L}-2 \boldsymbol{E}_1-\boldsymbol{E}_2\right),~~~~  \frac{\boldsymbol{Z}_{3,2}}{w_{3,2}} =2 (\boldsymbol{S}-\boldsymbol{E}_1),~~~~\frac{\boldsymbol{Z}_{3,3}}{w_{3,3}} = 2 \boldsymbol{E}_2 \\
		\frac{\boldsymbol{Z}_{4,1}}{w_{4,1}}  &=2 \left(\boldsymbol{S}-\boldsymbol{E}_1-\frac{1}{2} \boldsymbol{E}_3\right),~~~~ \frac{\boldsymbol{Z}_{4,2}}{w_{4,2}} = 2 (\boldsymbol{E}_1 - \boldsymbol{E}_2).
	\end{split}
	\end{align}
The above data is sufficient to compute the pushforwards of the intersection products in \cref{eqn:Y4intersection}. Following the approach of \cite{Jefferson:2022xft}, we organize the intersection numbers in \cref{eqn:X4intersection} into the generating function
	\begin{align}
	\label{eqn:genfun1}
		\exp\bigl({ \sum \phi_I \hat D_I}\bigr) &= \exp\bigl({ \sum \phi_I \boldsymbol{D}_I}\bigr)  \boldsymbol{X}_4 
	\end{align}
where $\phi_I$ in the above expression are formal parameters. Using the algorithm and accompanying \emph{Mathematica} package of \cite{Jefferson:2022xft}, we easily compute the pushforward of the generating function of the intersection numbers in \cref{eqn:genfun1} with respect to the composition of maps in \cref{eqn:SU5pushmap}:
\begin{align}
\begin{split}
\label{eqn:genfun}
Z&=\frac{e^{\sum \phi_\alpha D_\alpha}}{6 L  (L-S)^2}  \biggl(e^{- \phi _0 L}  (-L S (e^{(\phi _3-\phi _2) (3 L-2 S)+\phi _0 L + \phi _2 S})+6 L^2  e^{\phi _0 L +S (\phi _3-\phi _2)+ \phi _2 S}\\
&-6 \phi _1L^2  S  e^{(\phi _0  +\phi _3) L }-6 \phi _4 L^2  S  e^{(\phi _0+ \phi _3)L}+6\phi _3 L^2  S   e^{(\phi _0 + \phi _3)L}-6 L^2\\
&-6\phi _3 L S^2  e^{( \phi _0+ \phi _3)L}+6 \phi _1 L S^2  e^{(\phi _0+ \phi _3)L}+6 \phi _4 L  S^2  e^{( \phi _0+ \phi _3)L}+6 S^2  e^{( \phi _0+ \phi _3)L}\\
&-5 L  S e^{\phi _0 L + (\phi _3-\phi _2)S+ \phi _2S}-3 L  S  e^{(2 \phi _1-\phi _2) (S-L)+( \phi _0+ \phi _3)L}\\
&-3 L  S  e^{(2 \phi _4-\phi _3)  (S-L)+(\phi _0  +\phi _3)L }+12 L  S-6 S^2)\biggr).
\end{split}
\end{align}
Observe that $Z$ only depends non-trivially on the classes of the divisors $L = L_\alpha D_\alpha, S = S_\alpha D_\alpha \subset B$. The intersection numbers of divisors can be recovered from the above generating function by taking derivatives:
	\begin{align}
	\label{eqn:intnum}
		\hat D_{I_1} \cdots \hat D_{I_p} =  \left. \frac{\partial^p}{\partial \phi_{I_1} \cdots \partial \phi_{I_p}} Z\right|_{\phi_I =0}.
	\end{align}

\subsection{String theory cross checks}
\label{sec:crosscheck}

In this subsection, we perform a number of cross checks of \cref{eqn:genfun} in order to verify the validity of \cref{eqn:main,eqn:maincor}. These cross checks are motivated largely, although not entirely, by the physical interpretation of the intersection numbers of divisors $\hat D_I$ of CY spaces $X$ as coupling constants for topological terms in the low energy effective actions describing F/M theory compactifications on $X$. At the heart of this physical interpretation is the thoroughly tested conjecture that F-theory compactified on $X \times S^1$, where $X$ is an elliptically-fibered CY variety, is dual to M-theory compactified on $X$. In the M-theory duality frame, gauge fields in various dimensions can be obtained by dimensionally-reducing the three-form of eleven-dimensional supergravity on $(1,1)$-forms Poincar\'e dual to divisors of $X$, 
	\begin{align}
	\label{eqn:dimred}
		C_3 = \sum A^I \wedge \text{PD}(\hat D_I) + \cdots 
	\end{align}
We focus mostly on compactifications on CY 3-folds and 4-folds, and thus the cross checks we perform are based on the physical interpretation of, specifically, triple and quadruple intersection numbers. Some of the physical interpretations we mention below are discussed in further detail in the introduction of \cite{Jefferson:2022xft} (see also Appendix B of \cite{Jefferson:2021bid} and references therein).\newline

\noindent \textbf{Check \#0: Inverse Killing form from double intersections of Cartan divisors.} Before proceeding to more involved computations, one of the simplest cross checks entails using the intersection numbers to recover the inverse Killing form of the Lie algebra of SU(5).\footnote{This cross check actually applies for all F-theory compactifications on resolved SU(5) models defined by elliptically-fibered CY $n$-folds with $n\geq 2$.} Recall (see \cref{Kodaira}) that the irreducible components (i.e. rational curves) of the singular elliptic fibers over a generic point in $S$ intersect in the pattern of an affine SU(5) Dynkin diagram, where the nodes and edges of the Dynkin diagram correspond to the irreducible components and their points of pairwise intersection, respectively. Recall also that the intersection structure of a Dynkin diagram admits an equivalent presentation as a Cartan matrix, which is in turn related to the inverse Killing form by multiplication by a diagonal matrix.\footnote{Note that for simply-laced Lie algebras, the inverse Killing form is equal to the Cartan matrix, i.e. the diagonal matrix relating the two is simply the identity matrix.} It follows that the inverse Killing form matrix is encoded in the geometry of singular Kodaira fibers. Moreover, since the Cartan divisors $\hat D_i$ are topologically $\mathbb P^1$ fibrations over $S$, it is perhaps unsurprising that inverse Killing form is also visible in the intersection structure of Cartan divisors. In particular, one can verify that
	\begin{align}
		 \hat D_{i}  \hat D_j=  \left.\frac{\partial}{\partial \phi_i} \frac{\partial}{\partial \phi_j} Z\right|_{\phi_I=0} =- S\left(
\begin{array}{cccc}
 2 & -1 & 0 & 0 \\
 -1 & 2 & -1 & 0 \\
 0 & -1 & 2 & -1 \\
 0 & 0 & -1 & 2 \\
\end{array}
\right),
	\end{align}
precisely as expected for any resolution of the SU(5) model. This expectation is consistent with the classification results of Kodaira and N\'eron referenced in \cref{Kodaira}, as well as the physical interpretation of SU(5) as the gauge group of the low energy effective theory describing the dynamics of seven branes wrapping $S$ in F-theory vacua defined by the singular SU(5) model.\newline

\noindent \textbf{Check \#1: Five-dimensional Chern-Simons couplings from triple intersections.} The next cross check utilizes two results that have been discussed extensively in the physics literature. 

The first of these two results specifies the precise manner in which triple intersections of divisors of a CY 3-fold change under a flop transition. This can be stated informally as follows: Suppose $X$ and $X'$ are two smooth CY 3-folds related by a flop transition\footnote{By ``flop transition'', we more generally mean a topology-changing transition that consists of contracting a finite set of isolated exceptional curves in $X$, each with normal bundle $\mathcal O(-1) \oplus \mathcal O(-1)$, and blowing up a similar set of curves in $X'$. Note that the Kähler cones of $X,X'$ share a common wall in the extended Kähler cone.} characterized by the contraction of a curve $C \subset X$ and the blowup of a curve $C' \subset X'$. Then, the triple intersection numbers and intersection between divisors and second Chern class satisfy \cite{wilson1999flops,mcduff1994j}
	\begin{align}
	\begin{split}
	\label{eqn:flop}
	 \hat D_{I_1}'  \hat D_{I_2}'  \hat D_{I_3}' &=	\hat D_{I_1}  \hat D_{I_2}  \hat D_{I_3}-  (\hat D_{I_1}  C )( \hat D_{I_2}  C )( \hat D_{I_3}  C).
	\end{split}
	\end{align}
The formula in \cref{eqn:flop} implies
	\begin{align}
	\label{eqn:prepflop}
		\text{vol}(X')=  \text{vol}(X)- \frac{1}{3!} \text{vol}(C)^3,
	\end{align}
where
\begin{align}
	\label{eqn:prepotential}
		\text{vol}(X) := \frac{1}{3!} \sum \phi_I \phi_J \phi_K \hat D_I  \hat D_J  \hat D_K,~~~~ \text{vol}(C): =\sum \phi_I \hat D_I  C
	\end{align}
are the calibrated volumes of $X,C$. 

The second result concerns the correspondence between the kinematics of five-dimensional supergravity and the geometry of smooth CY 3-folds. M-theory compactified on a smooth CY 3-fold $X$ is expected to be described at low energies by a five-dimensional minimally-supersymmetric (i.e. eight supercharges) supergravity describing the dynamics of a collection of vector multiplets and hypermultiplets \cite{cadavid199511,Antoniadis:1995vz,Ferrara:1996hh,Witten:1996qb}. The couplings characterizing the massless two derivative action are encoded entirely in a one-loop exact homogeneous polynomial $F$ called the prepotential \cite{Seiberg:1996bd}, which corresponds to the polynomial of triple intersection numbers of $X$ (see, e.g., \cite{Intriligator:1997pq}):
	\begin{align}
		 F \longleftrightarrow \text{vol}(X).
	\end{align}
Although generically the low energy effective theory describing M-theory on a smooth 3-fold $X$ only consists of abelian vector multiplets, for certain theories the CY 3-fold background admits singular limits in which the abelian gauge symmetry enhances to a nonabelian gauge symmetry, 
	\begin{align}
		\text{U}(1)^{\text{rank}(\textbf G)} \to \textbf G.
	\end{align}
 In such cases, the abelian gauge theory corresponding to M-theory on smooth $X$ is in the Coulomb phase, and the kinetic and cubic couplings of the theory are understood to incorporate one-loop corrections generated by integrating out massive fermions on the Coulomb branch. 
 
 It is illuminating to compare this ``top down'' perspective with general results applicable to five-dimensional supersymmetric gauge theories. Given a five-dimensional minimally-supersymmetric gauge theory with gauge symmetry $\textbf G$ and collection of hypermultiplets transforming in the representation $\textbf R$, the pure gauge contributions $ F_{\text{Cartan}}$ to the one-loop exact prepotential 
 	\begin{align}
		 F = \cdots +  F_{\text{Cartan}}
	\end{align}
take the following form on the Coulomb branch \cite{Witten:1996qb,Intriligator:1997pq,Grimm:2013oga,Bonetti:2013ela}:
	\begin{align}
	\begin{split}
	\label{eqn:FCartan}
		 F_{\text{Cartan}} =\frac{1}{3!}k_{ijk} \phi_i \phi_j \phi_k = -\frac{1}{3!} \frac{1}{2} \sum_{\textbf R} N_{\textbf R} \sum_{w \in \textbf R} \text{sign}(Z_w) Z_w^3,~~~~ Z_w := \phi_i w_i,
	\end{split}
	\end{align}
where the $k_{ijk}$ are Chern-Simons coupling constants that arise in the M-theory vacuum as triple intersection numbers,
	\begin{align}
	\label{eqn:CScouplings}
		k_{ijk} \longleftrightarrow \hat D_i  \hat D_j  \hat D_k,
	\end{align}	
 $N_{\textbf R}$ is the number of hypermultiplets transforming in the representation $\textbf R$, and $w_i$ are the Dynkin coefficients of the weight $w \in \textbf R$.\footnote{Not to be confused with the weights of weighted projective spaces!} Thus, reversing the correspondence between geometry and physics, we can use \cref{eqn:FCartan} and related properties of five-dimensional supersymmetric gauge theories to ``predict'' topological triple intersection numbers of Cartan divisors of some as-yet-undetermined CY 3-fold.

We now turn to the task of checking \cref{eqn:genfun} using the two results described above. Let $X_4$ be a smooth CY 3-fold (i.e. an elliptic fibration over a 2-fold $B$) and consider the triple intersections of Cartan divisors,
	\begin{align}
		\hat D_i  \hat D_j  \hat D_k = \left.  \frac{\partial}{\partial \phi_i} \frac{\partial}{\partial \phi_j} \frac{\partial }{\partial \phi_k} Z \right|_{\phi_I =0}.
	\end{align}
The fact that $X_4$ is in particular an elliptically-fibered CY 3-fold admitting a singular ``F-theory'' limit means that much of the kinematic structure of the gauge sector of the five-dimensional theory is inherited from the tensor branch description of the six-dimensional lift describing the F-theory limit at low energies (see, e.g., \cite{Esole:2015xfa}). We use properties of the gauge sector of the low energy five-dimensional theory to independently compute the triple intersection numbers of $X_4$ in terms of the intersection numbers of the resolution $\mathscr B_{1,3}$, which have been computed geometrically using the standard pushforward formulae of \cite{Esole:2017kyr} (as opposed to the weighted blowup generalization described in this paper) and checked extensively against known physical properties of the corresponding low energy theory \cite{Esole:2014hya}.  We accomplish this task by combining \cref{eqn:flop} and \cref{eqn:FCartan} in a manner that we now describe in detail.

To begin, note that our resolution $X_4$ (see \cref{eqn:SU5resolution}) corresponds to $\mathscr B_{2,3}^1$ in Figure 5 of  \cite{Esole:2014hya}.  Note also that in Figure 5 of \cite{Esole:2014hya}, the two resolutions $\mathscr B_{2,3}^{1}$ and $\mathscr B_{1,3}$ are separated by a third resolution $\mathscr B_{2,3}$, whose triple intersection numbers we need not compute geometrically. Assuming that the stated correspondence between five-dimensional supersymmetric gauge theory and CY 3-fold geometry holds, we can use \cref{eqn:FCartan} to compute the triple intersection numbers of Cartan divisors $\hat D_i \subset \mathscr B_{2,3}^1$. Although \cref{eqn:FCartan} cannot be used to recover all of the triple intersection numbers of $\mathscr B_{2,3}^1$, for our purposes the differences $ F_{\mathscr B_{2,3}^1} -  F_{\mathscr B_{2,3}}$ and $ F_{\mathscr B_{2,3}} -  F_{\mathscr B_{1,3}}$ will suffice, and these differences only depend on triple intersection numbers of Cartan divisors. 

Let us elaborate on why the differences $ F_{\mathscr B_{2,3}^1} -  F_{\mathscr B_{2,3}}$ and $ F_{\mathscr B_{2,3}} -  F_{\mathscr B_{1,3}}$ only depend on Cartan divisors. The reason for this is that a flop transition $X \to X'$ can be viewed as a transition in which the contracting curve $C$ shrinks to zero volume and formally continues to negative volume, 
	\begin{align}
	\label{eqn:flopvol}
		\text{vol}(C) \rightarrow - \text{vol}(C),~~~~ \text{vol}(C) =\sum \phi_I \hat D_I  C.
	\end{align}
On the other hand, we note that $Z_w$ in \cref{eqn:FCartan} is the Coulomb branch central charge of a BPS particle corresponding to an M2 branes wrapping a holomorphic or antiholomorphic curve $C_w$.\footnote{That is, a curve $C_w$ such that the Dynkin coefficients $w_i$ of the weight $w \in \textbf R$ satisfy $w_i =- \hat D_i  C_w$.} In conventional six-dimensional F-theory compactifications (further reduced on a circle to five dimensions and dualized to M-theory), such curves only intersect the Cartan divisors and thus we are free to make the identification \cite{Intriligator:1997pq}
	\begin{align}
		Z_w \longleftrightarrow -\text{vol}(C_{w}).
	\end{align} 
The above property is true for the resolutions we are currently considering. Thus, the formal transition in \cref{eqn:flopvol} from positive to negative volume has a low energy manifestation as a singular phase transition characterized by a particle becoming massless. At negative volume, it appears as though the central charge has flipped sign, so we can simply flip the signs on some of the central charges on the right-hand side of \cref{eqn:FCartan} to determine the triple intersection numbers in a resolution related by a composition of flops. 

In our specific example, we would like to track the triple intersections of Cartan divisors in $\mathscr B_{1,3}$ through the following composition of flops:
	\begin{align}
		\mathscr B_{1,3} \overset{C_{w_6^{\textbf{10}}}}{\rightarrow} \mathscr B_{2,3} \overset{C_{w_3^{\textbf{10}}}}{\rightarrow} \mathscr B_{2,3}^1 = X_4.
	\end{align}
The relevant curves were identified in \cite{Esole:2017kyr} as corresponding to the following weights of the two-index antisymmetric representation $\textbf{10}$ of SU(5):
	\begin{align}
		w_6^{\textbf{10}} =( -1,1,-1,1),~~~~ w_3^{\textbf{10}} = (1,0,-1,1). 
	\end{align}
The upshot of this discussion is that the five-dimensional physical interpretation of the triple intersection numbers in \cref{eqn:FCartan,eqn:CScouplings} implies 
	\begin{align}
	\begin{split}
	\label{eqn:diff}
		 \text{vol}(X_4) -  \text{vol}(\mathscr B_{1,3}) &=(\text{vol}(\mathscr B_{2,3}^1) - \text{vol}(\mathscr B_{2,3}) )+ (\text{vol}(\mathscr B_{2,3})- \text{vol}(\mathscr B_{1,3}))\\
		 &=  \frac{1}{3!} N_{\textbf{10}}((Z_{w_{3}^{\textbf{10}}})^3 +  (Z_{w_6^{\textbf{10}}})^3)\\
		 &=  \frac{1}{3!} N_{\textbf{10}} (( - \phi_1 + \phi_2 - \phi_3 + \phi_4)^3 + (\phi_1 - \phi_3 + \phi_4)^3)
	\end{split}
	\end{align}
where $N_{\textbf{10}} = L  S$ and in going from the first line to the second line above, we have implicitly made use of \cref{eqn:FCartan}. Indeed, defining
	\begin{align}
	\begin{split}
		6\text{vol}(X_4) &= \left. \frac{\partial}{\partial \phi_I} \frac{\partial}{\partial \phi_J} \frac{\partial}{\partial \phi_K} Z \right|_{\phi_M =0} \phi_I \phi_J \phi_K
	\end{split}
	\end{align}
and using the analogous expression for $6 \text{vol}(\mathscr B_{1,3}) $ calculated using the methods of \cite{Jefferson:2022xft},
we find that the difference $\text{vol}(X_4) - \text{vol}(\mathscr B_{1,3})$ precisely matches \cref{eqn:diff}, as expected. \newline

\noindent \textbf{Check \#2: Three-dimensional Chern-Simons couplings from quadruple intersections.} Our final physics cross check of the intersection numbers \cref{eqn:genfun} is similar to the previous cross check, with the main difference being that we now take $X_4$ to be a CY 4-fold, i.e. an elliptic fibration over a smooth 3-fold base $B$. The low energy effective description of M-theory compactified on a smooth CY 4-fold $X$ is a three-dimensional supergravity theory with four supercharges---see \cite{Jefferson:2021bid} and references therein for detailed background relevant to the following discussion. 

When the M-theory vacuum solution in addition includes an appropriately quantized\footnote{The precise quantization condition was argued in \cite{Witten:1996md} to be $G_4 + \frac{1}{2} c_2 \in H^4(X,\mathbb Z)$. In this paper, we ignore any subtleties that could occur when $c_2$ is not an even class.} background profile for field strength of the M-theory three form,
	\begin{align}
	\label{eqn:G4}
		G_4 = \mathrm{d}C_3,
	\end{align}
 with $G_4$ belonging to the ``vertical'' subgroup of middle cohomology \cite{Greene:1993vm}, $H_{\text{vert}}^{2,2}(X,\mathbb R) \cap H^4(X,\mathbb Z)$,  the low energy effective three-dimensional theory contains Chern-Simons terms with couplings \cite{Grimm:2011sk,Cvetic:2012xn}
 	\begin{align}
	\label{eqn:3DCS}
		\Theta_{IJ} \longleftrightarrow- \int_X G_4 \wedge \text{PD}(\hat D_I) \wedge \text{PD}(\hat D_J),~~~~ \text{PD}(\hat D_I) \in H^{1,1}(X).
	\end{align}
(Note that in the above equation ``PD'' denotes the Poincar\'e dual.) The relationship between the Chern-Simons couplings $\Theta^{\text{3D}}_{IJ}$ and the intersection numbers of $X$ can be seen as follows. We first expand $G_4$ in a basis of vertical $(2,2)$-forms,
	\begin{align}
	\label{eqn:verticalG}
		G_4 = \sum G_{4,IJ} \text{PD}(\hat D_I) \wedge \text{PD}(\hat D_J),
	\end{align}
Then, roughly speaking, the right-hand side of \cref{eqn:3DCS} can then be expressed as a sum over quadruple intersection numbers:
	\begin{align}
	\begin{split}
	\label{eqn:fluxintegral}
		\int_X G_4 \wedge \text{PD}(\hat D_I) \wedge \text{PD}(\hat D_J)&= \sum G_{4,KL} \hat D_I  \hat D_J  \hat D_{K}  \hat D_L. 
	\end{split}
	\end{align}
As in the previous cross check, in order to make contact with the geometry of our resolution $X_4$ of the SU(5) model, we again narrow our discussion to special cases in which $X$ is an elliptically-fibered CY 4-fold admitting a singular F-theory limit. Let us also assume that the low energy four-dimensional minimally-supersymmetric (i.e. four supercharges) description of the F-theory limit is characterized by a gauge symmetry $\textbf G$ and some numbers $n_{\textbf{R}}, n_{\textbf{R}^*}$ of vector-like pairs of chiral multiplets transforming in the complex representations $\textbf R, \textbf R^*$. In order to preserve four-dimensional Poincar\'e symmetry and gauge symmetry $\textbf G$, it is necessary to restrict the background magnetic flux profile to specific values that are compatible with the F-theory lift \cite{Dasgupta:1999ss,Donagi:2008ca},
	\begin{align}
	\label{eqn:Flift}
	 G_4 \rightarrow \tilde G_4,~~~~ \int_{X} \tilde G_4 \wedge \text{PD}(\hat D_I) \wedge \text{PD}(\hat D_\alpha) = 0.
	\end{align} 
Such a flux profile is known to generically produce an excess of chiral multiplets in the four-dimensional spectrum \cite{Donagi:2008ca,Braun:2011zm,Marsano:2011hv,Krause:2011xj,Grimm:2011fx}, 
	\begin{align}
	\label{eqn:chiralindex}
		 n_{\textbf R} - n_{\textbf R^*} =: \chi_{\textbf R} = \sum\chi_{\textbf R,ij} \Theta_{ij}.
	\end{align}
The chiral excess $\chi_{\textbf R}$, called the `chiral index', is the index of a certain Dirac operator characterizing the dynamics of the seven brane worldvolume, and is hence topological. This is consistent with the dimensional reduction in \cref{eqn:dimred}, which suggests that the seven brane world-volume effective action is modified by the presence of a nontrivial magnetic gauge flux background along the internal compact directions. 

To make the connection with CY geometry, we observe that the three-dimensional one-loop Chern-Simons couplings can be computed directly in the low energy three-dimensional supersymmetric field theory \cite{Cvetic:2012xn} (see also \cite{Jefferson:2021bid}, Eq. (5.4)):
	\begin{align}
	\label{eqn:3DCScoupling}
		\Theta_{ij}= -\frac{1}{2} n_{\textbf R}\sum_{w \in \{ \textbf R, \textbf R^*\}} \text{sign}(Z_w) w_i w_j, ~~~~Z_w = \phi_i w_i. 
	\end{align}
In well-behaved constructions, combining \cref{eqn:chiralindex,eqn:3DCS} enables us to express the chiral index as a linear combination of fluxes
	\begin{align}
	\label{eqn:3DCSmatch}
	\begin{split}
		\chi_{\textbf R} =\sum \chi_{\textbf{R},{ij}} \tilde G_{4,KL} \hat D_i  \hat D_j  \hat D_K  \hat D_L.
	\end{split}
	\end{align}	

Let us carry out the above computation for the resolution $X_4$ of the SU(5) model obtained by the sequence of blowups in \cref{eqn:SU5resolution}. We first compute the quadruple intersection numbers, which we use as a basis for vertical fluxes in \cref{eqn:fluxintegral},
	\begin{align}
		G_{4,KL}\left. \frac{\partial}{\partial \phi_I}  \frac{\partial}{\partial \phi_J}  \frac{\partial}{\partial \phi_K}  \frac{\partial}{\partial \phi_L} Z \right|_{\phi_M = 0}.
	\end{align}
Next, we impose the symmetry-preserving restrictions \cref{eqn:Flift} on $G_4$ to obtain a geometric expression for the three-dimensional Chern-Simons couplings, which we then compare with the low corresponding field-theoretic expressions computed using \cref{eqn:3DCScoupling}. Matching the two sets of expressions, we find 
	\begin{align}
		\begin{split}
		\label{eqn:3DCSsys}
			\Theta_{11} =\Theta_{22}&=\chi_{\textbf{5}}- \chi_{\textbf{10}} =- 2 c_{ij} \tilde G_4^{ij} L  S  (6L - 5 S) \\
			\Theta_{12}=\Theta_{23}&=\frac{1}{2} (\chi_{\textbf{10} } -\chi_{\textbf{5}}) = c_{ij} \tilde G_4^{ij} L  S  (6L - 5 S) \\
			\Theta_{33}& = - \chi_{\textbf{10}} = - 2 c_{ij} \tilde G_4^{ij} L  S  (6L - 5 S) \\
			\Theta_{34} = -\frac{1}{2} \Theta_{44} &= \frac{1}{2} (\chi_{\textbf 5} + 3 \chi_{\textbf{10}}) =c_{ij} \tilde G_4^{ij} L  S  (6L - 5 S)
				\end{split}
	\end{align}
where 
	\begin{align}
		\begin{split}
		c_{ij} \tilde G_{4,ij} &=\frac{1}{5}(2 \tilde G_{4,44} - \tilde G_{4,34} + \tilde G_{4,33} - \tilde G_{4,23} + 2 \tilde G_{4,22} - \tilde G_{4,12} + 2 \tilde G_{4,11}).
		\end{split}
	\end{align}
The linear system in \cref{eqn:3DCSsys} cannot admit a nontrivial solution unless the chiral excesses satisfy the constraint
	\begin{align}
		\chi_{\textbf{5}} + \chi_{\textbf{10}} = 0,
	\end{align}
which is precisely the condition necessary to eliminate gauge anomalies from the four-dimensional effective theory. Assuming that string theory compactifications are automatically self-consistent, it is thus to be expected that the spectrum of the low energy four-dimensional theory is anomaly-free, and therefore the above computation serves as yet another cross check of the intersection numbers in \cref{eqn:genfun}. 

\section{Four-Dimensional F-theory Flux Vacua}
\label{sec:4DFlux}
\subsection{Vertical intersection pairing}
Emboldened by the success of our cross checks of \cref{eqn:main,eqn:maincor} in \cref{sec:P2,sec:SU5}, we next set our sights on novel applications relevant to four-dimensional string compactifications. 

One of the major physics motivations for generalizing the pushforward formulae of \cite{Esole:2017kyr} to be able to accommodate weighted blowups originates in the study of four-dimensional F-theory vacua with flux, which represent an attractive class of constructions in which supersymmetric extensions of the Standard Model with chiral spectra can be concretely realized \cite{Donagi:2008ca,Beasley:2008dc,Beasley:2008kw}. (The four-dimensional SU(5) model flux compactification in the previous section can be regarded as a candidate stringy construction of an SU(5) Grand Unified Theory.)

Compared to six-dimensional F-theory vacua, four-dimensional F-theory vacua are significantly more complicated. One reason for this is geometric in nature: as the dimension of the base increases, it becomes possible for the elliptic fibers to develop richer singularity structure over higher codimension components of the discriminant locus $V(\Delta)$, and the calculations required to study the geometry of these singularities can quickly become intractable. Moreover, in contrast to the classification of Kodaira singular fibers over codimension-one components of the discriminant locus \cite{KodairaII,KodairaIII,NeronClassification}, there is no rigorous classification for collisions of Kodaira fibers appearing over higher-codimension loci.\footnote{See \cite{Hayashi:2014kca,Esole:2015sta,Esole:2016npg} for physics-inspired classifications of codimension-two singular fibers, based on the low energy gauge theory interpretation of singular elliptically-fibered CY $n$-folds.}

For certain physics questions, the comparative lack of results on higher-codimension singular fibers is not a major obstruction. For example, questions related to the gauge symmetries and localized matter typically only require one to analyze the geometry of singular fibers over loci of codimension one and two, respectively. However, the Yukawa interactions and chiral spectrum in general depend sensitively on the  geometry of codimension-three singular fibers, and thus for these types of physics questions higher-codimension singularities cannot be ignored. This is where computations can become difficult, as analyzing their geometric properties often necessitates performing additional crepant blowups of the CY variety centered at subloci of relatively high codimension. Unfortunately, as the number and intricacy of the blowups required to read off the desired physics increases, the set of possible (globally) crepant blowups rapidly shrinks. Thus, in many circumstances, the generalization to weighted crepant blowups appears to be a necessity for maintaining access to a suitably broad set of possible blowups. 

Another reason that four-dimensional F-theory vacua are more complicated than their six-dimensional cousins is that four-dimensional F-theory vacua preserve half as many supercharges.  Consequently, much less of the physics is encoded simply in the geometry of the underlying CY variety, and additional mathematical structures must be considered in order to study certain physical properties. For example, as discussed in \cref{sec:crosscheck}, four-dimensional F-theory vacua with chiral spectra can be defined by a choice of elliptically-fibered CY variety and vertical flux background corresponding to the F-theory lift of the M-theory four-form field strength---see \cref{eqn:G4} and below. In order to compute the chiral spectrum, one needs to study integrals of the form \cref{eqn:fluxintegral}, which highlights the vertical intersection pairing restricted to the vertical cohomology,
	\begin{align}
		Q_{\text{vert}} := \left. Q \right|_{H^{2,2}_{\text{vert}}(X,\mathbb R) \cap H^4(X,\mathbb Z)},~~~~Q ~:~ H^{4}(X,\mathbb Z) \times H^{4}(X,\mathbb Z) \rightarrow \mathbb Z,
	\end{align}
as an additional mathematical structure of interest. In practice, it is possible to compute the matrix elements of $Q_{\text{vert}}$ in terms of the quadruple intersection numbers of divisors, which we can see appear on the right-hand side of \cref{eqn:fluxintegral} ---the key insight is that the vertical intersection pairing $Q_{\text{vert}}$ is (at least for an integral sublattice of $H^{2,2}_{\text{vert}}(X,\mathbb R) \cap H^4(X,\mathbb Z)$) the non-degenerate part of the generically singular matrix of quadruple intersection numbers \cite{Jefferson:2021bid}
	\begin{align}
	\label{eqn:Mmatrix}
		M_{(IJ)(KL)} := 	(\hat D_I \hat D_J) (\hat D_K \hat D_L).
	\end{align}
This is where the results of \cref{sec:push} are indispensable, as they provide a convenient means to compute the intersection numbers of weighted blowups of CY varieties, and hence to analyze the chiral matter multiplicities encoded in the pairing $Q_{\text{vert}}$. 

The goal of this section is to initiate the study of the vertical intersection pairing in weighted blowups of elliptic CY 4-folds defining four-dimensional F-theory vacua with flux, using \cref{eqn:main,eqn:maincor}. Although the full scope of physical phenomena (in particular, chiral physics) associated to codimension-three singular fibers in singular elliptic CY 4-folds is not fully understood at present, in \cite{Jefferson:2021bid}, it was emphasized that resolutions of F-theory models with $(4,6)$ loci\footnote{A $(4,6)$ locus is a sublocus of the discriminant locus $\Delta = 0$ in $B$ over which the singular elliptic fibers are characterized by $(f,g,\Delta)$ vanishing to respective orders greater than or equal to $(4,6,12)$ \cite{Candelas:2000nc}.} in the base can lead to non-flat fibrations\footnote{In this context, a `non-flat fibration' is a resolution of a singular elliptic fibration in which the dimension of some of the irreducible components of certain singular fibers increases over the $(4,6)$ locus in $V(\Delta) \subset B$. Although resolutions of singular four-dimensional F-theory models leading to non-flat fibrations have been known about for quite some time \cite{Candelas:2000nc,Hayashi:2009ge,achmed2018note}, their physics has remained unclear.}  in which irreducible surface components of the fibers can be threaded non-trivially by vertical background flux profiles. This in turn points to the intriguing possibility that chiral degrees of freedom in the four-dimensional spectrum may reside at $(4,6)$ loci in some models.\footnote{The only compact simple Lie groups admitting complex (and hence, chiral) matter representations are $\text{SU}(N), \text{SO}(4k+2), \text{E}_6$ with $N>4$ and $k>1$.} However, the analysis of \cite{Jefferson:2021bid} did not work out the additional blowups necessary to fully explore the chiral physics of $(4,6)$ loci in certain models with residual codimension-three singularities, among them being the F$_4$ and Sp(6) models. This omission was partly due to the lack of available mathematical techniques; we take steps to fill this gap in the following discussion, deferring more thorough physical analyses to future work \cite{JLT}.

\subsection{F$_4$ model}

The F$_4$ model, characterized by a type IV$^{*\text{ns}}$ fiber over a divisor $S  = V(s) \subset B$, is defined by imposing the following vanishing behavior on the sections $a_i$ in the Weierstrass equation in \cref{eqn:Tateform} \cite{Esole:2017kyr,Esole:2017rgz}:
	\begin{align}
		a_1 = 0,~~~~ a_2=0,~~~~ a_3 = 0,~~~~ a_4 = s^{3+n} a_{4,3+n} ,~~~~ a_6 =s^4 a_{6,4},
	\end{align}
for some choice of $n \in \mathbb Z_{>0}$. For simplicity, we restrict to the case $n = 0$. The sections $f,g,\Delta$ take the following form:
	\begin{align}	
		\begin{split}
		\label{eqn:F4fgD}
		f = a_{4,3} s^3,~~~~g= a_{6,4} s^4,~~~~\Delta = ( 4 a_{4,3}^3 s+ 27 a_{6,4}^2 ) s^8. 
		\end{split}
	\end{align}
This model can be resolved (through singular elliptic fibers appearing over generic points in codimension-two components of the discriminant locus) by a sequence of four ordinary blowups,
	\begin{align}
		\label{eqn:F4resolutioncodim2}
			 f_1 = (x,y,s|e_1),~~f_2 = (y,e_1|e_2),~~ f_3=(x,e_2|e_3) ,~~ f_4 =  (e_2,e_3|e_4).
	\end{align}
From \cref{eqn:F4fgD}, we see that we can further adjust the parameters of this model to exhibit $(4,6)$ loci along
\begin{align}
	\label{eqn:F446}
		V(s , a_{4,3} , a_{6,4}) \subset B.
	\end{align}
 by demanding $(a_{4,3}, a_{6,4})$ to vanish to orders $(1,2)$ along additional loci that intersect $V(s) \subset B$. 
Note that there are still residual singular fibers over the above locus in the F$_4$ model after the four blowups in \cref{eqn:F4resolutioncodim2}. Although we do not attempt to classify all residual singularities here, we note that the cases in which $(a_{4,3}, a_{6,4})$ vanish to orders $(0,1),(1,0),(1,1)$---which are not $(4,6)$ loci---do not lead to additional singular elliptic fibers. Therefore, we focus our attention on the cases described immediately below \cref{eqn:F446}.

\subsubsection{$(4,6)$ curve}
First, let us consider the case of a codimension-two $(4,6)$ locus, given by constraining the orders of vanishing of $a_{4,3}, a_{6,4}$ in the vicinity of an independent divisor $T= V(t) \subset B$:
	\begin{align}
		a_{4,3} =t \alpha,~~~~~ a_{6,4} = t^2\beta.
	\end{align}
We now have a $(4,6)$ curve $S \cap T \subset B$. With this additional condition, the F$_4$ model after the four blowups (\ref{eqn:F4resolutioncodim2}) has a non-empty singular locus, $V(t,e_3, y)$. In this case, we can perform an additional weighted crepant blowup of the ambient space centered precisely at the singular locus,
	\begin{align}
	\label{eqn:fifthblowup}
		f_5 : Y_5 \rightarrow Y_4,~~~~f_5 = (t,e_3,y|e_5). 
	\end{align}
The exceptional divisor $E_5 \subset Y_5$ is a 4-fold with the topology of a 3-fold fibration, $ \mathbb P( \mathcal O_{\mathbb P^1} \oplus \mathcal O_{\mathbb P^1} \oplus \mathcal O_{\mathbb P^1} (-2) ) \rightarrow  S \cap T$, with fiber homogeneous coordinates
	\begin{align}
	\label{eqn:E5coord}
		{E}_5~:~ [ t: \tilde e_3: \tilde y]~,~ [\tilde x : \tilde e_4],
	\end{align}
Thus, the elliptic CY 4-fold has acquired a new divisor $\hat D_5 = {E}_5 \cap {X}_5$, 
whose fibers are surface components of a non-flat fibration where each surface component has the topology of a conic fibration over $\mathbb P^1$; here, the conic fibers correspond to the following vanishing locus in the $\mathbb P^2$ subspace of the fibers of $Y_5$, with homogeneous coordinates $[t:\tilde e_3: \tilde y]$:
	\begin{align}
		\hat D_5 ~:~\begin{bmatrix} t \\ \tilde e_3 \\ \tilde y \end{bmatrix}^{\mathrm t} \begin{bmatrix} 2 \beta&\alpha \tilde e_4 \tilde x & 0 \\ \alpha \tilde e_4 \tilde x & 2 \tilde e_4 \tilde x^3 & 0 \\ 0 & 0 &-2 \end{bmatrix}  \begin{bmatrix} t \\ \tilde e_3 \\ \tilde y \end{bmatrix} = 0.
	\end{align}

What can this geometry teach us about the physics of four-dimensional flux compactifications on the F$_4$ model? To answer this question, it is most convenient to work in the basis of divisors
basis of divisor classes
	\begin{align}
		\begin{split}
		\label{eqn:F4Dbasiscurve}
			\boldsymbol D_0 &= \frac{1}{3} \boldsymbol H \\
			\boldsymbol  D_1 &= \boldsymbol E_1 - \boldsymbol E_2 \\
			\boldsymbol D_2 &= -\boldsymbol E_1 + 2 \boldsymbol E_2 - \boldsymbol E_3 - \boldsymbol E_4\\
			\boldsymbol D_3 &= \boldsymbol E_1 - 2 \boldsymbol E_2 + \boldsymbol E_3 + 2 \boldsymbol E_4 \\
			\boldsymbol D_4 &= \boldsymbol E_3 - \boldsymbol E_4 -\boldsymbol E_5 \\
			\boldsymbol D_5 & = \boldsymbol E_5,
		\end{split}
	\end{align}
where we note 
	\begin{align}
	 \boldsymbol{X}_5 = 3\boldsymbol{H} + 6\boldsymbol{L}- \boldsymbol{E}_1 - \boldsymbol{E}_5 - \sum_{i=1}^5 \boldsymbol{E}_i.
	\end{align}
Notice that the exceptional divisor $\hat D_5$ intersects the Cartan divisors $\hat D_3, \hat D_4$ in new vertical surfaces, as can be seen by intersecting $\hat D_5$ with the hyperplanes (resp.) $V(\tilde e_4), V(\tilde e_3)$:
	\begin{align}
	\begin{split}
	\label{eqn:codim2newsurface}
		 \hat D_3 \cap \hat D_5 ~&:~ V( \tilde y^2 - \beta \tilde t^2 )~\subset ~ {E}_5 \cap {V(\tilde e_4)}\\
		\hat D_4 \cap \hat D_5 ~&:~V( \tilde y^2 -\beta \tilde t^2 )~\subset ~{E}_5 \cap {V(\tilde e_3 )},
	\end{split}
	\end{align}
where the coordinates of the fibers of ${E}_5$ are given in \cref{eqn:E5coord}. 

In order to understand in detail how the presence of these vertical surfaces affects the lattice of vertical fluxes, we need to compute the intersection pairing restricted to the vertical subgroup of middle cohomology. Our first task is to compute the intersection numbers. Following the same approach as in \cref{sec:SU5int}, we compute the pushforward of the generating function, 
$$\varpi_* \circ f_{1*} \circ f_{2*} \circ f_{3*} \circ f_{4*} \circ f_{5*}\exp(\sum \phi_I \boldsymbol D_I) \boldsymbol X_5,$$
 in terms of the basis of divisors in \cref{eqn:F4Dbasiscurve}, using the standard version of the pushforward formulae in \cref{sec:push} (i.e. with trivial weights, $w_i =1$)---this computation can be carried out easily using the algorithm and accompanying software package developed in \cite{Jefferson:2022xft}. For the purpose of studying vertical flux backgrounds in M-theory compactifications on $X_5$, i.e. flux backgrounds $G_4$ of the form displayed in \cref{eqn:verticalG}, we then organize the quadruple intersection numbers of $X_5$ into the matrix in \cref{eqn:Mmatrix} and determine its nondegenerate part in a reduced basis,
	\begin{align}
	\label{eqn:F4m}
		\left[ \begin{array}{cc} m &\vec v\\
		\vec v^{\mathrm t} & 4 S T (L-S)
		\end{array}\right],~~~~IJ = \alpha 0, \alpha \beta, \alpha i , 35,
	\end{align}
where
	\begin{align}
	\begin{split}
	\label{eqn:F4M}
		m &=  D_\alpha D_{\alpha'}\left(  \begin{bmatrix} - L & D_\beta \\
		 D_{\beta'} &0 \end{bmatrix} \oplus (-S)\begin{bmatrix} 2 & -1 & 0& 0&0 \\ -1 & 2 & -2 &0&0 \\0 & -2 &4 &-2 &0\\ 0 &0& -2 &4&0\\0&0&0&0&0\end{bmatrix}  \right),~~~~
		 \vec v = 2 D_\alpha S T \begin{bmatrix}0 \\0\\ 0\\0\\ 0 \\ 1 \\ -1\end{bmatrix}.
	\end{split}
	\end{align}
While the submatrix $m$ was already implicitly computed from the four blowup resolution of the F$_4$ model in \cite{Jefferson:2021bid}, the gluing vector $\vec v$ and lower right matrix element $4ST (L-S)$ are new elements of the reduced vertical pairing matrix introduced by $f_5$, and correspond to intersection pairings involving the new vertical surfaces in \cref{eqn:codim2newsurface}. It is evident from the above intersection pairing matrix for vertical four-form fluxes that the additional basis element $\hat D_3 \hat D_5$ introduced by the fifth blowup $f_5$ vanishes when the intersection between the base divisors trivializes, $S T =0$, since this condition forces the lower-right matrix element $4 ST (L-S)$ and lattice gluing vector $\vec v$ to vanish. If we apply the usual constraints necessary for the flux background $G_4$ to lift to a Poincar\'e invariant four-dimensional F-theory vacuum with unbroken $\text{F}_4$ gauge symmetry, we find that restricted matrix of quadruple intersections in \cref{eqn:F4m} takes the following form:
	\begin{align}
		 2 ST \begin{bmatrix} 0 & - D_\alpha  \\ - D_{\alpha'}  & 2 (L-S)\end{bmatrix},~~~~IJ = \alpha 5, 35
	\end{align}
Evidently, the fifth blowup resolving the codimension-two $(4,6)$ locus has introduced new, independent vertical surfaces to the sublattice of symmetry-preserving vertical flux backgrounds, with homological representatives
	\begin{align}
		[\hat D_\alpha \cap \hat D_5], ~~~~[ \hat D_3 \cap \hat D_5].
	\end{align} 
We defer a more thorough analysis of the physics associated to the new divisor $\hat D_5$ to future work.

\subsubsection{$(4,6)$ point}
So far, we have not required the generalization of the pushforward formulae described in \cref{sec:push}. However, in this next case, we find that \cref{eqn:main,eqn:maincor} are necessary to accommodate weighted blowups of the F$_4$ model. Consider the case of a codimension-three $(4,6)$ locus, obtained by imposing the following orders of vanishing on $a_{4,3}, a_{6,4}$ in the vicinity of a pair of divisors $V(t) =T, V(u) = U \subset B$:
	\begin{align}
		a_{4,3} = t \alpha , ~~~~ a_{6,4} = t u \beta.
	\end{align}
Our model now has a codimension-three $(4,6)$ point $S\cap T \cap U \subset B$. This particular realization of the F$_4$ model after the four blowups in \cref{eqn:F4resolutioncodim2} still has a non-empty singular locus $V( t ,u ,e_3, y )$. In contrast to the previous example, this time we perform a \emph{weighted} blowup of the ambient space $Y_4$ restricting to a crepant blowup of $X_4$, centered at the singular locus,
	\begin{align}
	\label{eqn:fifthF4codim3}
		f_5 : Y_5 \rightarrow Y_4, ~~~~f_5 = (t,u,e_3,y|e_5)_{\frac{1}{2},\frac{1}{2},\frac{1}{2}, \frac{1}{2}}.
	\end{align}
The new ``Cartan'' divisor $\hat D_5 = {E}_5 \cap  Y_5$ of $X_5 \subset Y_5$ is again a 3-fold, given by the vanishing locus
	\begin{align}
	V(	\tilde y^2 - ( \tilde e_3^2 \tilde e_4 \tilde x^3 + \alpha \tilde e_3 \tilde e_4 t \tilde x + \beta t u) ) \subset E_5,
	\end{align}
 which has the topology of $\mathbb P( \mathcal O_{\mathbb P^1} \oplus \mathcal O_{\mathbb P^1}  \oplus \mathcal O_{\mathbb P^1}  \oplus \mathcal O_{\mathbb P^1} (-2) )  \rightarrow S\cap T\cap U$. The fibers of the exceptional divisor have homogeneous coordinates
	\begin{align}
		{E}_5~:~[ t:u:\tilde y: \tilde e_3]~,~ [ \tilde x: \tilde e_4]. 	
	\end{align}
The intersections of the new divisor $\hat D_5$ with the Cartan divisors $\hat D_3, \hat D_4$ (respectively, the hyperplanes $V(\tilde e_4 ), V(\tilde e_3) \subset Y_5$) are new, distinct vertical surfaces:
	\begin{align}
		\begin{split}
			\hat D_3 \cap \hat D_5 ~&:~ V(\tilde y^2 - \beta tu )~ \subset~ {E}_5 \cap {V(\tilde e_4)}\\
			\hat D_4 \cap \hat D_5~&:~ V(\tilde y^2 - \beta t u)~ \subset ~ {E}_5 \cap {V(\tilde e_3)}.
		\end{split}	
	\end{align}
Notice that the surface $\hat D_3  \cap \hat D_5$ has a rational double point singularity along $V(t , u ,\tilde y )$, meaning that our elliptic CY 4-fold $X_5$ still contains residual singularities even after five blowups. In this paper, we limit our analysis to the proper transform $X_5$, as this residual singularity should not affect the qualitative results we wish to illustrate---we leave a complete resolution of this model to future work.

To better understand the significance of the surfaces $\hat D_3 \cap \hat D_4,\hat D_4 \cap \hat D_5$ to four-dimensional flux vacua, we repeat the analysis of the previous subsection. An important distinction in this case is that we are required to use the generalized pushforward formulae of \cref{sec:push} to compute the intersection numbers, due to the fact that the fifth blowup in \cref{eqn:fifthF4codim3} is a weighted blowup. In terms of the basis of divisor classes
	\begin{align}
		\begin{split}
		\label{eqn:F4Dbasis}
			\boldsymbol D_0 &= \frac{1}{3} \boldsymbol H \\
			\boldsymbol  D_1 &= \boldsymbol E_1 - \boldsymbol E_2 \\
			\boldsymbol D_2 &= -\boldsymbol E_1 + 2 \boldsymbol E_2 - \boldsymbol E_3 - \boldsymbol E_4\\
			\boldsymbol D_3 &= \boldsymbol E_1 - 2 \boldsymbol E_2 + \boldsymbol E_3 + 2 \boldsymbol E_4 \\
			\boldsymbol D_4 &=\frac{1}{2} (2 \boldsymbol E_3 -2 \boldsymbol E_4 -\boldsymbol E_5) \\
			\boldsymbol D_5 & = \boldsymbol E_5,
		\end{split}
	\end{align}
with
	\begin{align}
		\boldsymbol{X}_5 = 3\boldsymbol{H} + 6\boldsymbol{L}- \boldsymbol{E}_1 - \sum_{i=1}^5 \boldsymbol{E}_i,
	\end{align}
we find that the nondegenerate part of the matrix of quadruple intersections in \cref{eqn:Mmatrix} can be written in the reduced basis 
\begin{align}
IJ = \alpha 0, \alpha \beta ,\alpha i, 34,44.
\end{align}
The matrix of quadruple intersections is similar to the previous case in \cref{eqn:F4m}. Restricting the lattice of vertical flux backgrounds to those that preserve the full F$_4$ gauge symmetry in the F-theory lift, we obtain the following restricted matrix of quadruple intersections
	\begin{align}
	\label{eqn:restrictedm}
		2 S T U \begin{bmatrix}0  & -1 \\-1 & 6 \end{bmatrix},~~~~	IJ = 34,44.
	\end{align}
Thus we see that the case where the F$_4$ model has a $(4,6)$ point along $S\cap T\cap U \subset B$ is qualitatively distinct from the previous case of a $(4,6)$ curve $ST \subset B$: evidently, the fifth blowup introduces just two new vertical surfaces with homology classes\footnote{Note that the vertical surface Chow ring classes $\hat D_3 \hat D_5, \hat D_4 \hat D_5$ are related to $\hat D_3 \hat D_4, \hat D_4 \hat D_4$ by the linear relations encoded in the null vectors of the matrix of quadruple intersections in \cref{eqn:Mmatrix}; this explains the difference in basis between \cref{eqn:restrictedm} and \cref{eqn:restrictedbasis}.}
	\begin{align}
	\label{eqn:restrictedbasis}
		[\hat D_3 \cap \hat D_5], ~~~~ [\hat D_4 \cap \hat D_5]. 
	\end{align}
The physics of similar F$_4$ models with $(4,6)$ points will be explored more extensively in forthcoming work \cite{JKT}.

\subsection{Sp(6) model}

The Sp(6) model, characterized by a type I$_{6}^{\text{ns}}$ Kodaira fiber over a divisor $S \subset B$, is defined by imposing the following vanishing behavior on the sections $a_i$ in the Weierstrass equation in \cref{eqn:Tateform} \cite{Esole:2017kyr}:
	\begin{align}
		a_1 = 0,~~~~ a_2 = a_2,~~~~ a_3 =0,~~~~ a_4 = s^6 a_{4,6},~~~~ a_6 = s^{12} a_{6,12}.
	\end{align}
The sections $f,g, \Delta$ take the following form:
	\begin{align}
	\label{eqn:Sp6fgD}
		 f = - \frac{1}{3} a_2,~~~~ g= \frac{2}{27} a_2^3,~~~~ \Delta = a_2^2 s^{12} ( -a_{4,6}^2 + 4 a_2 a_{6,12} )
	\end{align}
This model can be resolved through codimension two by a sequence of six unweighted blowups, with blowdown maps 
	\begin{align}
		f_1 \circ \cdots \circ f_6 : Y_6 \rightarrow Y_0,
	\end{align}
	 where 
	\begin{align}
	\label{eqn:Sp6sixblowups}
		f_1= (x,y,s|e_1),~~~~f_2 = (x,y,e_1|e_2), ~~~~ f_3 = (x,y,e_2|e_3)  ~~~~\dots~~~~f_6 = (x,y,e_5|e_6). 
	\end{align}
From \cref{eqn:Sp6fgD}, we see that we can further restrict this model to exhibit $(4,6)$ loci by demanding that $a_2$ vanish to order at least two along some locus intersecting $V(s)$.
Again, we do not attempt to classify all residual singularities in the Sp(6) model after the six blowups. Nevertheless, similar to the case of the F$_4$ model, by taking $a_2$ to vanish along an independent divisor in $B$, we find that the Sp(6) model does not develop additional singularities unless $a_2$ vanishes to at least order two, and hence we again focus our attention on $(4,6)$ singularities in the following discussion.

	\subsubsection{$(4,6)$ curve}
	
	Let $T = V(t) \subset B$, and consider the case of a $(4,6)$ curve $S \cap T \subset B$, which we obtain by imposing the following vanishing behavior:
		\begin{align}
			a_2 = \alpha t^2. 
		\end{align}	
	With this additional condition, the Sp(6) model after the six blowups in \cref{eqn:Sp6sixblowups} now exhibits additional singularities along the reducible locus
	\begin{align}
		V(e_2 e_3 e_4 , t ,y ) \subset B.
	\end{align}	
We resolve these singularities by means of a sequence of three additional unweighted blowups
	\begin{align}
\label{eqn:F6codim2blowups}
	 	f_7 \circ f_8 \circ f_9 : Y_9 \rightarrow Y_6,~~~~f_7 = (e_2 , t, y| e_7),~~~~ f_8 = (e_3,t, y|e_8),~~~~ f_9 = (e_4, t,y |e_9).
	\end{align}
There are now three exceptional divisors of the ambient space $Y_9$ to consider, with the following topologies:
	\begin{align}
		\begin{split}
		{E}_7 ~&\cong ~ \text{Bl}_{V(\tilde e_3 , \tilde t ,\tilde y )} \,\mathbb P(\mathcal O_{\mathbb P^1} \oplus \mathcal O_{\mathbb P^1} \oplus \mathcal O_{\mathbb P^1}(-2) )  \rightarrow S\cap T\\
			{E}_8 ~&\cong  ~\text{Bl}_{V( \tilde e_4 , \tilde t , \tilde y)} \,\mathbb P(\mathcal O_{\mathbb P^1} \oplus \mathcal O_{\mathbb P^1} \oplus \mathcal O_{\mathbb P^1}(-2) ) \rightarrow S\cap T\\
			{E}_9 ~&\cong ~ ~~~~~~~~~~~~\,\, \mathbb P(\mathcal O_{\mathbb P^1} \oplus \mathcal O_{\mathbb P^1} \oplus \mathcal O_{\mathbb P^1}(-2)) \rightarrow S\cap T,
		\end{split}	
	\end{align}
where the fiber coordinates of the above exceptional divisors are
	\begin{align}
		\begin{split}
			{E}_7~&: ~[\tilde t: \tilde y: \tilde e_3 ] ~,~ [ e_8 \tilde t: e_8 \tilde y: \tilde e_2 ] ~,~[  e_8\tilde e_3: \tilde e_1] \\
			{E}_8~&:~  [\tilde t: \tilde y: \tilde e_4 ] ~,~ [ e_9 \tilde t: e_9 \tilde y: \tilde e_3] ~,~[ e_9 \tilde e_4 : \tilde e_7]\\
			{E}_9 ~&:~ [ \tilde t: \tilde y: \tilde e_4] ~,~ [\tilde e_5 :\tilde e_8].
				\end{split}
	\end{align}
	Intersecting the above exceptional divisors with the CY hypersurface $ X_9$ (i.e. the proper transform under the three blowups indicated in \cref{eqn:F6codim2blowups}), we see that the three additional blowups have produced the three divisors
	\begin{align}
		\hat D_7 = {E}_7 \cap {X}_9,~~~~ \hat D_8 = {E}_8 \cap {X}_9,~~~~ \hat D_9 = {E}_9 \cap {X}_9.
	\end{align}
Each of the divisors $\hat D_7, \hat D_8, \hat D_9$ has the topology of a surface fibered over $S\cap T \subset B$, where the fibers are described by the following equations:	
	\begin{align}
	\begin{split}
		\hat D_7 ~&:~V(\tilde y^2 - \alpha \tilde t^2 - \tilde e_1\tilde e_2^2 \tilde e_3^3 e_8 ) ~\subset~{E}_7\\
		\hat D_8 ~&:~ V( \tilde y^2 - \alpha \tilde t^2)~ \subset ~ {E}_8 \\
		\hat D_9 ~&:~V( \tilde y^2 - \alpha \tilde t^2 - a_{4,6} \tilde e_4^2 \tilde e_5 \tilde e_8) ~\subset~ {E}_9.
	\end{split}
	\end{align}
We next study the vertical intersection pairing for $X_9$, by working out the nondegenerate part of the matrix of quadruple intersection numbers. Let us work in the basis of divisors
	\begin{align}
		\begin{split}
					\boldsymbol{D}_1 &= \boldsymbol{E}_1 - \boldsymbol{E}_2 \\
					\boldsymbol{D}_2 &= \boldsymbol{E}_2 - \boldsymbol{E}_3 - \boldsymbol{E}_7\\
					\boldsymbol{D}_3 &= \boldsymbol{E}_3 - \boldsymbol{E}_4 - \boldsymbol{E}_8\\
					\boldsymbol{D}_4 &= \boldsymbol{E}_4 - \boldsymbol{E}_5 -\boldsymbol{E}_9\\
					\boldsymbol{D}_5 &= \boldsymbol{E}_5 - \boldsymbol{E}_6 \\
					\boldsymbol{D}_{i = 6,\dots ,9} &= \boldsymbol{E}_{i = 6,\dots,9},
		\end{split}
	\end{align}
with
	\begin{align}
		\boldsymbol{X}_9 &= 3 \boldsymbol{H} + 6 \boldsymbol{L} - 2\sum_{i=1}^{9} \boldsymbol E_i.
	\end{align}
(Note that although these divisors could themselves still contain additional residual singularities of $X_9$, they do not affect the qualitative nature of the present discussion and hence we need not worry about them.) In terms of the above basis, we observe that a number of new vertical surfaces have been introduced,
	\begin{align}
		\begin{split}
			&\hat D_{1} \cap \hat D_7,~~~~\hat D_{2}\cap \hat D_7,~~~~\hat D_{3}\cap \hat D_7\\
			&\hat D_{3}\cap \hat D_8,~~~~ \hat D_{4}\cap \hat D_8 ,~~~~\hat D_{7} \cap\hat D_8\\
			&\hat D_{4}\cap \hat D_9,~~~~\hat D_{5}\cap \hat D_9 , ~~~~\hat D_8\cap \hat D_9,
		\end{split}
	\end{align}
possibly not all independent in homology. We next compute the pushforward of the generating function $\exp(\sum \phi_I \boldsymbol{D}_I) \cap \boldsymbol X_9$ with respect to the composition of maps $\varpi \circ f_1 \circ \cdots \circ f_9 : Y_9 \rightarrow B$. The  pushforward of this function can again be easily evaluated using the algorithm and \emph{Mathematica} package of \cite{Jefferson:2022xft}. Since the resulting expression is rather cumbersome, we do not include it here. 

With the generating function of intersection numbers in hand, we then proceed to compute a matrix representation of the nondegenerate part of the matrix of quadruple intersection numbers in \cref{eqn:Mmatrix}:
	\begin{align}
		\begin{bmatrix} m & \vec v_{37} & \vec v_{48} & \vec v_{59} \\
		\vec v_{37}^{\mathrm t} & -2 S^2 T & & \\
		\vec v_{48}^{\mathrm t} &&-2 S^2 T& \\
		\vec v_{59}^{\mathrm t} &&& - 4 S T (L- 2 S)
		\end{bmatrix},~~~~IJ = \alpha 0, \alpha \beta ,\alpha i, 37,48,59,
	\end{align}
where 
	\begin{align}
		m =  D_\alpha D_{\alpha'}\left(  \begin{bmatrix} - L & D_\beta \\
		 D_{\beta'} &0 \end{bmatrix} \oplus (-S)\begin{bmatrix}
\begin{array}{ccccccccc}
 4 & -2 & 0 & 0 & 0 & 0 & 0 & 0 & 0 \\
 -2 & 4 & -2 & 0 & 0 & 0 & 0 & 0 & 0 \\
 0 & -2 & 4 & -2 & 0 & 0 & 0 & 0 & 0 \\
 0 & 0 & -2 & 4 & -2 & 0 & 0 & 0 & 0 \\
 0 & 0 & 0 & -2 & 4 & -2 & 0 & 0 & 0 \\
 0 & 0 & 0 & 0 & -2 & 2 & 0 & 0 & 0 \\
 0 & 0 & 0 & 0 & 0 & 0 & 0 & 0 & 0 \\
 0 & 0 & 0 & 0 & 0 & 0 & 0 & 0 & 0 \\
 0 & 0 & 0 & 0 & 0 & 0 & 0 & 0 & 0 \\
\end{array}
\end{bmatrix}  \right)
	\end{align}
and the gluing vectors associated to the lower right $3 \times 3$ diagonal block matrix are
	\begin{align}
		\vec v_{37} = D_\alpha S T \begin{bmatrix} 0\\0\\0\\2\\-2\\0\\0\\0\\-2\\2\\0\end{bmatrix},~~~~\vec v_{48} = D_\alpha S T \begin{bmatrix} 0\\0\\0\\0\\2\\-2\\0\\0\\0\\-2\\2\end{bmatrix},~~~~ \vec v_{59} = D_\alpha S T \begin{bmatrix} 0\\0\\0\\0\\0\\2\\0\\0\\0\\0\\-2\end{bmatrix}.
	\end{align}
We next consider the nondegenerate part of the matrix of quadruple intersections restricted to a sublattice of M-theory vertical flux background that to F-theory flux vacua with unbroken Sp(6) gauge symmetry. We find that the in the restricted basis $IJ = \alpha 7, \alpha 8, \alpha 9,37,48,59$, the matrix of quadruple intersections takes the form
	\begin{align}
	\tilde m =	
			2 S T \begin{bmatrix}0  &\vec u_{37} & \vec u_{48} & \vec u_{59} \\
		\vec u_{37}^{\mathrm t} &T-S & 0&-T \\
		\vec u_{48}^{\mathrm t} &0&T-S&-T \\
		\vec u_{59}^{\mathrm t} &-T&-T& 2 (-L + 2 S + 2 T)
		\end{bmatrix}
	\end{align}
with ``restricted'' gluing vectors
	\begin{align}
		\vec u_{37}=\begin{bmatrix}
			-D_\alpha \\
			D_\alpha \\ 0
		\end{bmatrix},~~~~ \vec u_{48} &=\begin{bmatrix}
		0\\
			-D_\alpha \\
			D_\alpha 
		\end{bmatrix},~~~~\vec u_{59} =\begin{bmatrix}
			0\\
			0 \\ -D_\alpha
		\end{bmatrix}.
	\end{align}
Evidently, the three additional blowups in \cref{eqn:F6codim2blowups} have introduced the vertical cycles
	\begin{align}
	\begin{split}
		&[\hat D_\alpha \cap \hat D_7], ~~~~ [\hat D_3 \cap \hat D_7]\\
		&[\hat D_\alpha \cap \hat D_8], ~~~~ [\hat D_4 \cap \hat D_8]\\
		&[\hat D_\alpha \cap \hat D_9], ~~~~ [\hat D_5 \cap  \hat D_9].
	\end{split}
	\end{align}
\subsubsection{$(4,6)$ point}
We next consider the case of $(4,6)$ point $S\cap T \cap U \subset B$, which we construct by imposing the following vanishing behavior:
	\begin{align}
		a_2 = \alpha  t u,
	\end{align}
where $V(t)=T, V(u )=U \subset B$. The Sp(6) model after the six blowups in \cref{eqn:Sp6fgD} now exhibits residual singularities along the reducible locus
	\begin{align}
	\label{eqn:singX6}
		V(e_2 e_3 e_4 , t ,u ,y) \subset Y_6.
	\end{align}
We attempt to resolve these singular fibers by means of a sequence of weighted blowups, beginning with the following two blowups:
	\begin{align}
	\label{eqn:Sp6extratwo}
		f_7 \circ f_8: Y_8 \rightarrow Y_6,~~~~ f_7 = (t,u,e_2,y|e_7)_{\frac{1}{2}, \frac{1}{2},\frac{1}{2},\frac{1}{2}},~~~~ f_8 = (t,u,e_4,y|e_8)_{\frac{1}{2},\frac{1}{2},\frac{1}{2},\frac{1}{2}}.
	\end{align}
After these two additional blowups, there is still a residual singularity of the proper transform $X_8$ along
	\begin{align}
		V(e_3, t , u , y )\subset Y_8.
	\end{align}	
Unfortunately, there do not seem to exist any weighted global crepant blowups centered at the above locus. To see this, suppose there exists a weighted blowup with blowdown map
	\begin{align}
		f_9 \overset{\text{?}}{=} (t,u,e_3,y|e_9)_{w_1,w_2,w_3,w_4},~~~~ w_{a=1,\dots,4} \in \mathbb Q_{>0}.
	\end{align}
We can explicitly check whether such a set of positive rational weights $w_a$ exists by imposing the following conditions:
	\begin{itemize}
		\item \emph{the blowup is crepant, i.e. the proper transform is obtained by factoring out $\sum_a w_a -1 $ copies of the exceptional locus;}
		\item \emph{and the section defining $X_9$ is well-defined, i.e. comprises monomials constructed from variables with integers powers}. 
	\end{itemize} 
To check these criteria, we use the fact that the total transform of $X_9$ is defined by the vanishing of the following polynomial:
	\begin{align}
		\begin{split}
			&e_9^{w_1 + w_2 + w_3 + w_4 -1} (-a_{4,6} e_1^5 e_2^4 e_3^3 e_4^2 e_5 e_7 s^6 e_9^{1-w_1-w_2-w_3+2w_4}x z^2 \\&-a_{6,12}e_1^{10} e_2^8 e_3^6 e_4^4 e_5^2 e_7^3 e_8  e_9^{1-w_1-w_2-w_3+5w_4}s^{12} z^3 -\alpha e_9^{1-w_1-w_4}  t u x^2 z\\
			&-e_1 e_2^2 e_3^3 e_4^4 e_5^5 e_6^6 e_8  e_9^{1-w_1-w_2-w_3+2w_4}x^3 +e_9^{1+w_1-w_2-w_3-w_4}y^2 z ).
		\end{split}
	\end{align}
This implies that the weights $w_a \in \mathbb Q_{>0}$ must also satisfy
	\begin{align}
		\begin{split}
			&w_1 + w_2 + w_3 + w_4 - 1 ,~~~~
			1 - w_1 - w_2 -w_3 + 2 w_4 , ~~~~1 - w_1 -w_2 -w_3 + 5 w_4 \\
		&1-w_1 - w_4 ,~~~~1+ w_1 - w_2 - w_3 -w_4~~ \in~~ \mathbb Z_{\geq 0}.
		\end{split}
	\end{align}
The first thing to notice is that these conditions imply that
	\begin{align}
		w_1 + w_4 = 1,
	\end{align}
which we can use to simplify the above system of inequalities. Defining $r_2 = -w_2-w_3+3 w_4$, the above system implies
	\begin{align}
		r_2 ,~~~~ 3 w_4 + r_2 ,~~~~ -5 w_4 + r_2 + 2 ~~\in ~~ \mathbb Z_{\geq 0},
	\end{align}
which is clearly impossible if $w_4 < 1$ since $\text{gcd}(3,-5) =1$. 

It is possible that there is a suitable alternative sequence of potentially weighted blowups that resolves the singular locus in \cref{eqn:singX6}, but for now we simply content ourselves to study the qualitative features of this realization of the Sp(6) model after the two additional blowups in \cref{eqn:Sp6extratwo}---namely, the appearance of additional vertical surfaces formed exclusively out of the pairwise intersections of ``Cartan'' divisors $\hat D_{j=1,\dots,8}$. Thus, we proceed with a provisional analysis of the partial resolution (i.e. eight blowups) of the Sp(6) model, keeping in mind that there are residual singularities that could change some details of the analysis once appropriately resolved. 

The two exceptional divisors of the ambient space $Y_8$ have the following topologies:
	\begin{align}
	\begin{split}
		{E}_7 ~&\cong ~ \mathbb P( \mathcal O_{\mathbb P^1} \oplus \mathcal O_{\mathbb P^1} \oplus \mathcal O_{\mathbb P^1} \oplus \mathcal O_{\mathbb P^1}(-2) )\rightarrow S \cap T \cap U\\
		{E}_8 ~ &\cong ~ \mathbb P( \mathcal O_{\mathbb P^1} \oplus \mathcal O_{\mathbb P^1} \oplus \mathcal O_{\mathbb P^1} \oplus \mathcal O_{\mathbb P^1}(-2) ) \rightarrow S\cap T \cap U , 
	\end{split}
	\end{align}
where the fiber homogeneous coordinates of the above exceptional divisors are 
	\begin{align}
		\begin{split}
			{E}_7 ~&:~ [ \tilde t: \tilde u : \tilde y : \tilde e_2] ~,~ [\tilde e_1 : \tilde e_3] \\
			{E}_8 ~&:~ [ \tilde t : \tilde u : \tilde y : \tilde e_4] ~,~ [ \tilde e_3 : \tilde e_5]. 
		\end{split}
	\end{align}
Intersecting the above exceptional divisors with the CY hypersurface $ {X}_8$ (i.e. the proper transform of the partial resolution $X_6$ of the Sp(6) model under the two additional blowups indicated in \cref{eqn:Sp6extratwo}), we see that the two additional blowups have produced the divisors 
	\begin{align}
		\hat D_7 = {E}_7 \cap  X_8,~~~~\hat D_8 = {E}_8 \cap  X_8.
	\end{align}
Both divisors $\hat D_7, \hat D_8$ have the topology of a 3-fold (over the point $S \cap T \cap U \subset B$), where the fiber components of the non-flat fibration are described by the following equations:
	\begin{align}
		\begin{split}
			\hat D_7 ~&:~ V(\tilde y^2 - \tilde e_1 \tilde e_2^2 \tilde e_3^3 -\alpha \tilde t \tilde u ) ~\subset~  E_7\\
			\hat D_8 ~&:~ V(\tilde y^2 - \alpha e_7 t u  - a_{4,6} \tilde e_3^3 \tilde e_4^2 \tilde e_5 ) ~ \subset ~  E_8. 
		\end{split}	
	\end{align}
We next study the vertical intersection pairing for $X_8$, following the same procedure as in the previous subsection, but this time using the basis 
	\begin{align}
				\begin{split}
					\boldsymbol{D}_1 &= \boldsymbol{E}_1 - \boldsymbol{E}_2 \\
					\boldsymbol{D}_2 &= \boldsymbol{E}_2 - \boldsymbol{E}_3 - \frac{1}{2} \boldsymbol{E}_7\\
					\boldsymbol{D}_3 &= \boldsymbol{E}_3 - \boldsymbol{E}_4\\
					\boldsymbol{D}_4 &= \boldsymbol{E}_4 - \boldsymbol{E}_5 - \frac{1}{2} \boldsymbol{E}_8\\
					\boldsymbol{D}_5 &= \boldsymbol{E}_5 - \boldsymbol{E}_6 \\
					\boldsymbol{D}_{i = 6,\dots ,8} &= \boldsymbol{E}_{i = 6,\dots,8},
		\end{split}
	\end{align}
with	
	\begin{align}
		\boldsymbol X_8 = 3 \boldsymbol H + 6 \boldsymbol L - 2 \sum_{i=1}^{6} \boldsymbol E_i - \boldsymbol{E}_7 -\boldsymbol{E}_8. 
	\end{align}
Note that in terms of the above basis, we find that a number of new vertical surfaces have been introduced:
	\begin{align}
		\begin{split}
			&\hat D_1\cap  \hat D_7,~~~~ \hat D_2 \cap \hat D_7,~~~~\hat D_3 \cap \hat D_7\\
			&\hat D_3 \cap  \hat D_8,~~~~\hat D_4 \cap \hat D_8,~~~~\hat D_5 \cap \hat D_8,
		\end{split}	
	\end{align}
where each of the divisors has the topology of a complex surface over the point $S\cap T \cap U \subset B$. We retrieve the quadruple intersection numbers from the pushforward of the generating function $\exp(\sum \phi_I \boldsymbol D_I)  \boldsymbol X_8$ with respect to the map $\varpi \circ f_1 \circ \cdots \circ f_8 :Y_8 \rightarrow B$, again applying the results of \cref{sec:push} to the \emph{Mathematica} package in \cite{Jefferson:2022xft}. We compute a representative of the nondegenerate part of the matrix of quadruple intersection numbers in the basis
	\begin{align}
	IJ = \alpha 0, \alpha \beta, \alpha i, 12,22,44,45.
	\end{align}
Restricting to the sublattice of M-theory flux backgrounds lifting to four-dimensional F-theory flux vacua preserving Sp(6) gauge symmetry, we find that the restricted matrix of quadruple intersection numbers takes the form
	\begin{align}
		2 ST U\left( \begin{bmatrix}  &-1 \\ -1 & 6 \end{bmatrix} \oplus \begin{bmatrix} 6 & -1 \\ -1 & \end{bmatrix} \right).
	\end{align}
Evidently, in this case, the additional two blowups in \cref{eqn:Sp6extratwo} produce new vertical surfaces with the following homological representatives in $H_{2,2}^{\text{vert}}(X_8)$:
	\begin{align}
		[\hat D_1 \cap \hat D_2],~~~~ [\hat D_2 \cap \hat D_2],~~~~ [ \hat D_4 \cap \hat D_4],~~~~[ \hat D_4 \cap \hat D_5]. 
	\end{align}
Although we anticipate that a third additional blowup (i.e. a ninth blowup), centered at $V( t , u ,e_3 , y)$, would introduce even more vertical surfaces; we leave the identification of a suitable blowup centered at this locus to future work.

\section{Discussion and Future Directions}

Motivated by computing topological intersection numbers of crepant resolutions of singular Calabi-Yau spaces relevant for string theory (in particular, F-theory) compactifications to various dimensions, we proved in \cref{main,maincor} a generalization of the pushforward formulae of \cite{Fullwood:2012kj,Esole:2017kyr}. The strategy for the proof of \cref{main} was to first prove the corresponding pushforward formula for a special case, namely a weighted blowup of affine space at the origin, and then show that the results pull back to the general case of a weighted blowup of an algebraic variety centered at a complete intersection of smooth hypersurfaces. The proof of the special case made use of the fact \cite{weighted} that a weighted blowup of the origin in affine space can be regarded, roughly, as an ordinary blowup of a particular sequence of root stacks followed by a sequence of ``de-rootings'' that leaves the exceptional divisor unchanged. 

These generalized pushforward formulae, given in \cref{eqn:main,eqn:maincor}, express the intersection products of the Chow ring class of the exceptional divisor of a weighted blowup of an arbitrary algebraic variety, centered at a complete intersection of smooth hypersurfaces intersecting locally transversally, in terms of intersection products of the classes of the hypersurfaces in the base of the blowup. We have extensively checked these generalized pushforward formulae by using them to compute numerous intersection numbers of weighted blowups and comparing the results to independent computations. In the process, we have performed the first known computation of the generating function of intersection numbers for weighted blowups of the SU(5) model defined over an arbitrary base. We have also performed the first known computation of the generating functions of intersection numbers of weighted blowups of the F$_4$ and Sp(6) models over arbitrary bases, constrained to exhibit so-called (4,6) elliptic fiber singularities over codimension-two and codimension-three components of the discriminant locus. We found that the crepant resolutions of the codimension-two and codimension-three non-minimal singularities in these cases led to nonflat fibrations in which some of the irreducible components of the singular elliptic fibers are divisors (i.e., codimension-one subvarieties in the blown-up Calabi-Yau space).

There are a number of interesting physical and mathematical future research directions facilitated or at least suggested by the results of this paper. We comment on a few of these directions:

\begin{itemize}
 \item \underline{Four-dimensional chiral matter at (4,6) points.} A conservative approach to four-dimensional F-theory model-building is to exclude Calabi-Yau backgrounds with (4,6) singularities. However, recent work has pointed out that in some cases the presence of (4,6) points does not necessarily spoil the desired phenomenology \cite{achmed2018note}. Moreover, it was emphasized in \cite{Jefferson:2021bid} that resolutions of (4,6) singularities lead to non-flat fibrations with elliptic fibers containing additional vertical surfaces supporting vertical four-form fluxes in M-theory, indicating the intriguing possibility that (4,6) points may  signal the presence of (possibly exotic) massless chiral matter in the F-theory limit. Ongoing \cite{JLT} and future work exploring this possibility may benefit substantially from the availability of mathematical techniques for analyzing weighted blowups resolving (4,6) loci.
 
 \item \underline{Exact quantization of M-theory four-form fluxes.} Another strong motivation for developing a better mathematical understanding of vertical four-form fluxes and their relation to (4,6) loci is developing a means to compute the exact (half) integer-quantized lattice of M-theory four-form flux backgrounds $H^{2,2}(X, \mathbb R) \cap H^4(X,\mathbb Z/2)$, which remains unknown for the vast majority of smooth Calabi-Yau fourfolds $X$ potentially relevant for Standard Model-like constructions in string theory. In \cite{Jefferson:2021bid}, it was conjectured that the intersection pairing restricted to the vertical subgroup of middle cohomology, which can be computed in terms of quadruple intersection numbers, is a resolution-invariant quantity that potentially sheds light on this problem. We anticipate that the mathematical technology developed in this paper, which provides a means to explicitly compute the intersection pairing associated to a much broader set of (4,6) singularities (which seem to be rather prevalent among four-dimensional F-theory vacua), will enable useful progress to be made towards this goal in both ongoing \cite{JKT} and future work. 

\item{} \underline{3-fold components of non-flat fibers.} Given the apparent necessity of weighted blowups for global crepant resolutions of higher codimension singularities in Calabi-Yau varieties, the computations performed in \cref{sec:4DFlux} in the context of the F$_4$ and Sp(6) models take a major first step towards a systematic exploration of (4,6) singularities in general four-dimensional F-theory vertical flux vacua. Although there are a plethora of unanswered questions regarding the physics of higher-codimension (4,6) singularities in four dimensional compactifications---both with and without fluxes---one immediate puzzle that has emerged from the analysis in \cref{sec:4DFlux} is the possibility that weighted blowups can lead to nonflat fibrations in which certain singular fibers contain various types of divisors (i.e., complex dimension is three) as irreducible components. It remains to be seen whether or not there is interesting physics directly associated to these divisor components (corresponding, e.g., to the physics of M2-branes, M5-branes, or Euclidean M5-instantons wrapping (resp.) holomorphic 2, 4, or 6-cycles), as well as whether such degrees of freedom lift to massless/tensionless degrees of freedom in the F-theory limit or are alternatively either obstructed by anomalies or lifted by some sort of mass-generation mechanism. 

\item{} \underline{Further generalization of the pushforward formulae in \cref{sec:push}.} The pushforward formulae proved in \cref{main,maincor} are proved for weighted blowups centered at a complete intersection of smooth hypersurfaces meeting locally transversally. It would be interesting to try and prove the vailidity of the resulting formulae for a more general class of blowups. 

\item{} \underline{Fractional weighted blowups.} Finally, although we have extensively used ``fractional'' weighted blowups as a means to identify crepant blowups in a larger class of examples, as mentioned in \cref{rem:frac} the underlying mathematics is perhaps not fully understood (at least to the authors), and may admit a more intrinsic and precise mathematical explanation in terms of root stacks. It would be interesting to determine to what extent there exists a mathematically precise definition of fractional weighted blowups, and if so, what its broader implications would be for the subject of weighted blowups. 

\end{itemize}

We hope to further explore the above topics in future work.

\label{sec:discussion}

\section*{Acknowledgements}
We are indebted to Dan Abramovich for his guidance and encouragement on this project, and we are extremely grateful to Paolo Aluffi for enabling this collaboration. We thank Dan Abramovich, Paolo Aluffi, and Andrew P. Turner for valuable comments on earlier versions of this paper. PJ thanks Manki Kim, Shing Yan Li, Paul-Konstantin Oehlmann, Washington Taylor, and Andrew P. Turner for useful discussions. VA and SO were supported in part by funding from BSF grant 2018193 and NSF grant DMS-2100548. PJ was supported by DOE grant DE-SC00012567. 

\appendix

\section{Crepant Weighted Blowup Condition}
\label{app:crepant}
Let $Z \subset Y$ be the complete intersection of $d$ nonsingular hypersurfaces $Z_1, \dots, Z_d$ meeting transversally in $Y$. Let $ f:  Y' \rightarrow Y$ be the weighted blowup of $Y$ centered at $Z$, with exceptional divisor $ E$ and weights $w_1, \dots,w_d$. The total Chern class of $ Y'$ is then \cite{article}
	\begin{align}
		c(T  Y') = (1+\boldsymbol E) \prod_{i=1}^d \frac{1 + \boldsymbol Z_i - w_i \boldsymbol E}{1 + \boldsymbol Z_i}c(TY),
	\end{align}
which implies 
	\begin{align}
		c_1(T Y') = (1-\sum_{i=1}^d w_i )\boldsymbol E + c_1(TY).
	\end{align}
Now suppose $ X \subset  Y$ is a complete intersection $X_1 \cap \cdots \cap X_p$, and let the proper transform $ X' \subset  Y'$ be the complete intersection $ X' =  X_1' \cap \cdots \cap  X_p'$, where $n_i$ copies of the exceptional divisor have been subtracted from the class of each hypersurface, i.e.,
	\begin{align}
		\boldsymbol{X}'= \prod_{i=1}^p \boldsymbol{X}_i'  = \prod_{i=1}^p (\boldsymbol X - n_i \boldsymbol E).
	\end{align}
	 The total Chern class of $ X'$ is given by
	\begin{align}
		c( X') = c(T X') \cap \boldsymbol{X}'.
	\end{align}	
Assume that $ X'$ is smooth. Using adjunction, we have
	\begin{align}
		c(T  X') = \frac{c(T  Y')}{c(N_{ X'/ Y'} )},~~~~ N_{ X'/ Y'} = \prod_{i=1}^p (1 + \boldsymbol{X}_i'),
	\end{align}
hence
	\begin{align}
		c_1(T  X') = c_1(T  Y') -  \sum_{i=1}^p \boldsymbol{X}_i'
	\end{align}
The first Chern class of $ X'$ vanishes iff the above expression is equal to zero, i.e. if
	\begin{align}
	\begin{split}
			0&= (\sum_{j=1}^p n_j - (\sum_{i=1}^{d} w_i-1))\boldsymbol{E} + c_1(TY) - \sum_{i=1}^p \boldsymbol{X}_i.
	\end{split}
	\end{align}
	If $X$ is Calabi-Yau, i.e., if $c_1(TY) = \sum_i \boldsymbol{X}_i$, then the above equation implies the blowup is crepant if and only if 
	\begin{align}
	\label{eqn:crepantcond}
		\sum_{i=1}^p n_i = \sum_{i=1}^d w_i -1. 
	\end{align}

\bibliographystyle{elsarticle-num}
\bibliography{references}

\begin{thebibliography}{10}
\expandafter\ifx\csname url\endcsname\relax
  \def\url#1{\texttt{#1}}\fi
\expandafter\ifx\csname urlprefix\endcsname\relax\def\urlprefix{URL }\fi
\expandafter\ifx\csname href\endcsname\relax
  \def\href#1#2{#2} \def\path#1{#1}\fi

\bibitem{JoyceCompact}
D.~Joyce, Compact manifolds with special holonomy, Oxford University Press,
  2000.

\bibitem{TianYau1}
G.~Tian, S.-T. Yau, {Complete Kahler manifolds with zero Ricci curvature I},
  Journal of the American Mathematical Society 3 3 (1990) 579--609.
\newblock \href {https://doi.org/10.2307/1990928} {\path{doi:10.2307/1990928}}.

\bibitem{TianYau2}
G.~Tian, S.-T. Yau, {Complete Kahler manifolds with zero Ricci curvature II},
  Inventiones mathematicae 106 (1991) 27--60.
\newblock \href {https://doi.org/10.1007/BF01243902}
  {\path{doi:10.1007/BF01243902}}.

\bibitem{GHJ}
M.~Gross, D.~Joyce, D.~Huybrechts, Calabi-Yau Manifolds and Related Geometries,
  Springer-Verlag, 2003.

\bibitem{Dixon:1985jw}
L.~J. Dixon, J.~A. Harvey, C.~Vafa, E.~Witten, {Strings on Orbifolds}, Nucl.
  Phys. B 261 (1985) 678--686.
\newblock \href {https://doi.org/10.1016/0550-3213(85)90593-0}
  {\path{doi:10.1016/0550-3213(85)90593-0}}.

\bibitem{Vafa:1997pm}
C.~Vafa, {Lectures on strings and dualities}, in: {ICTP Summer School in
  High-energy Physics and Cosmology}, 1997, pp. 66--119.
\newblock \href {http://arxiv.org/abs/hep-th/9702201}
  {\path{arXiv:hep-th/9702201}}.

\bibitem{Aspinwall:1994ev}
P.~S. Aspinwall, {Resolution of orbifold singularities in string theory},
  AMS/IP Stud. Adv. Math. 1 (1996) 355--379.
\newblock \href {http://arxiv.org/abs/hep-th/9403123}
  {\path{arXiv:hep-th/9403123}}.

\bibitem{becker_becker_schwarz_2006}
K.~Becker, M.~Becker, J.~H. Schwarz, String Theory and M-Theory: A Modern
  Introduction, Cambridge University Press, 2006.
\newblock \href {https://doi.org/10.1017/CBO9780511816086}
  {\path{doi:10.1017/CBO9780511816086}}.

\bibitem{reid1983minimal}
M.~Reid, et~al., Minimal models of canonical 3-folds, Adv. Stud. Pure Math 1
  (1983) 131--180.

\bibitem{Arras:2016evy}
P.~Arras, A.~Grassi, T.~Weigand, {Terminal Singularities, Milnor Numbers, and
  Matter in F-theory}, J. Geom. Phys. 123 (2018) 71--97.
\newblock \href {http://arxiv.org/abs/1612.05646} {\path{arXiv:1612.05646}},
  \href {https://doi.org/10.1016/j.geomphys.2017.09.001}
  {\path{doi:10.1016/j.geomphys.2017.09.001}}.

\bibitem{pub.1089196667}
M.~Reid,
  \href{https://app.dimensions.ai/details/publication/pub.1089196667}{Young
  person's guide to canonical singularities}, 1987, pp. 345--414.
\newblock \href {https://doi.org/10.1090/pspum/046.1/927963}
  {\path{doi:10.1090/pspum/046.1/927963}}.
\newline\urlprefix\url{https://app.dimensions.ai/details/publication/pub.1089196667}

\bibitem{10.2307/2152704}
J.~Kollar, S.~Mori, \href{http://www.jstor.org/stable/2152704}{Classification
  of three-dimensional flips}, Journal of the American Mathematical Society
  5~(3) (1992) 533--703.
\newline\urlprefix\url{http://www.jstor.org/stable/2152704}

\bibitem{Candelas:2000nc}
P.~Candelas, D.-E. Diaconescu, B.~Florea, D.~R. Morrison, G.~Rajesh,
  {Codimension three bundle singularities in F theory}, JHEP 06 (2002) 014.
\newblock \href {http://arxiv.org/abs/hep-th/0009228}
  {\path{arXiv:hep-th/0009228}}, \href
  {https://doi.org/10.1088/1126-6708/2002/06/014}
  {\path{doi:10.1088/1126-6708/2002/06/014}}.

\bibitem{Esole:2014hya}
M.~Esole, S.-H. Shao, S.-T. Yau, {Singularities and Gauge Theory Phases II},
  Adv. Theor. Math. Phys. 20 (2016) 683--749.
\newblock \href {http://arxiv.org/abs/1407.1867} {\path{arXiv:1407.1867}},
  \href {https://doi.org/10.4310/ATMP.2016.v20.n4.a2}
  {\path{doi:10.4310/ATMP.2016.v20.n4.a2}}.

\bibitem{Braun:2014kla}
A.~P. Braun, S.~Schafer-Nameki, {Box Graphs and Resolutions I}, Nucl. Phys. B
  905 (2016) 447--479.
\newblock \href {http://arxiv.org/abs/1407.3520} {\path{arXiv:1407.3520}},
  \href {https://doi.org/10.1016/j.nuclphysb.2016.02.002}
  {\path{doi:10.1016/j.nuclphysb.2016.02.002}}.

\bibitem{Jefferson:2022xft}
P.~Jefferson, A.~P. Turner, {Generating functions for intersection products of
  divisors in resolved F-theory models} (6 2022).
\newblock \href {http://arxiv.org/abs/2206.11527} {\path{arXiv:2206.11527}}.

\bibitem{fulton}
W.~Fulton, Intersection Theory, 2nd Edition, Springer New York, NY, 1998.

\bibitem{Fullwood:2011zb}
J.~Fullwood, {On generalized Sethi-Vafa-Witten formulas}, J. Math. Phys. 52
  (2011) 082304.
\newblock \href {http://arxiv.org/abs/1103.6066} {\path{arXiv:1103.6066}},
  \href {https://doi.org/10.1063/1.3628633} {\path{doi:10.1063/1.3628633}}.

\bibitem{Fullwood:2012kj}
J.~Fullwood, M.~van Hoeij, {On stringy invariants of GUT vacua}, Commun. Num.
  Theor Phys. 07 (2013) 551--579.
\newblock \href {http://arxiv.org/abs/1211.6077} {\path{arXiv:1211.6077}},
  \href {https://doi.org/10.4310/CNTP.2013.v7.n4.a1}
  {\path{doi:10.4310/CNTP.2013.v7.n4.a1}}.

\bibitem{Esole:2017kyr}
M.~Esole, P.~Jefferson, M.~J. Kang, {Euler Characteristics of Crepant
  Resolutions of Weierstrass Models}, Commun. Math. Phys. 371~(1) (2019)
  99--144.
\newblock \href {http://arxiv.org/abs/1703.00905} {\path{arXiv:1703.00905}},
  \href {https://doi.org/10.1007/s00220-019-03517-1}
  {\path{doi:10.1007/s00220-019-03517-1}}.

\bibitem{Esole:2018tuz}
M.~Esole, M.~J. Kang, {Characteristic numbers of crepant resolutions of
  Weierstrass models} (7 2018).
\newblock \href {http://arxiv.org/abs/1807.08755} {\path{arXiv:1807.08755}}.

\bibitem{Esole:2018bmf}
M.~Esole, M.~J. Kang, {Characteristic numbers of elliptic fibrations with
  non-trivial Mordell-Weil groups} (8 2018).
\newblock \href {http://arxiv.org/abs/1808.07054} {\path{arXiv:1808.07054}}.

\bibitem{Weigand:2018rez}
T.~Weigand, {F-theory}, PoS TASI2017 (2018) 016.
\newblock \href {http://arxiv.org/abs/1806.01854} {\path{arXiv:1806.01854}}.

\bibitem{Greene:1993vm}
B.~R. Greene, D.~R. Morrison, M.~R. Plesser, {Mirror manifolds in higher
  dimension}, Commun. Math. Phys. 173 (1995) 559--598.
\newblock \href {http://arxiv.org/abs/hep-th/9402119}
  {\path{arXiv:hep-th/9402119}}, \href {https://doi.org/10.1007/BF02101657}
  {\path{doi:10.1007/BF02101657}}.

\bibitem{Jefferson:2021bid}
P.~Jefferson, W.~Taylor, A.~P. Turner, {Chiral matter multiplicities and
  resolution-independent structure in 4D F-theory models} (8 2021).
\newblock \href {http://arxiv.org/abs/2108.07810} {\path{arXiv:2108.07810}}.

\bibitem{Katz:2011qp}
S.~Katz, D.~R. Morrison, S.~Schafer-Nameki, J.~Sully, {Tate's algorithm and
  F-theory}, JHEP 08 (2011) 094.
\newblock \href {http://arxiv.org/abs/1106.3854} {\path{arXiv:1106.3854}},
  \href {https://doi.org/10.1007/JHEP08(2011)094}
  {\path{doi:10.1007/JHEP08(2011)094}}.

\bibitem{Esole:2017rgz}
M.~Esole, P.~Jefferson, M.~J. Kang, {The Geometry of F$_4$-Models} (4 2017).
\newblock \href {http://arxiv.org/abs/1704.08251} {\path{arXiv:1704.08251}}.

\bibitem{JKT}
P.~Jefferson, M.~Kim, W.~Taylor, (in progress).

\bibitem{Candelas:1990rm}
P.~Candelas, X.~C. De~La~Ossa, P.~S. Green, L.~Parkes, {A Pair of Calabi-Yau
  manifolds as an exactly soluble superconformal theory}, Nucl. Phys. B 359
  (1991) 21--74.
\newblock \href {https://doi.org/10.1016/0550-3213(91)90292-6}
  {\path{doi:10.1016/0550-3213(91)90292-6}}.

\bibitem{hori2003mirror}
K.~Hori, C.~Vafa, C.~M. Institute, R.~Pandharipande,
  \href{https://books.google.com/books?id=uGYRaAhFGx0C}{Mirror Symmetry}, Clay
  mathematics monographs, American Mathematical Society, 2003.
\newline\urlprefix\url{https://books.google.com/books?id=uGYRaAhFGx0C}

\bibitem{7d3c9f5cc0d546b3a3c103df4d4cb4f5}
D.~Cox, J.~Little, H.~Schenck, Toric varieties, Vol. 124 of Graduate Studies in
  Mathematics, American Mathematical Society, United States, 2011.
\newblock \href {https://doi.org/10.1090/gsm/124} {\path{doi:10.1090/gsm/124}}.

\bibitem{weighted}
M.~H. Quek, D.~Rydh,
  \href{https://people.kth.se/~dary/weighted-blowups20220329.pdf}{Weighted
  blow-ups} (March 2022).
\newline\urlprefix\url{https://people.kth.se/~dary/weighted-blowups20220329.pdf}

\bibitem{EG}
D.~Edidin, W.~Graham, \href{https://doi.org/10.1007/s002220050214}{Equivariant
  intersection theory}, Invent. Math. 131~(3) (1998) 595--634.
\newblock \href {https://doi.org/10.1007/s002220050214}
  {\path{doi:10.1007/s002220050214}}.
\newline\urlprefix\url{https://doi.org/10.1007/s002220050214}

\bibitem{Vistoli:1989aa}
A.~Vistoli, \href{https://doi.org/10.1007/BF01388892}{Intersection theory on
  algebraic stacks and on their moduli spaces}, Inventiones mathematicae 97~(3)
  (1989) 613--670.
\newblock \href {https://doi.org/10.1007/BF01388892}
  {\path{doi:10.1007/BF01388892}}.
\newline\urlprefix\url{https://doi.org/10.1007/BF01388892}

\bibitem{Beasley:2008dc}
C.~Beasley, J.~J. Heckman, C.~Vafa, {GUTs and Exceptional Branes in F-theory -
  I}, JHEP 01 (2009) 058.
\newblock \href {http://arxiv.org/abs/0802.3391} {\path{arXiv:0802.3391}},
  \href {https://doi.org/10.1088/1126-6708/2009/01/058}
  {\path{doi:10.1088/1126-6708/2009/01/058}}.

\bibitem{Donagi:2008kj}
R.~Donagi, M.~Wijnholt, {Breaking GUT Groups in F-Theory}, Adv. Theor. Math.
  Phys. 15~(6) (2011) 1523--1603.
\newblock \href {http://arxiv.org/abs/0808.2223} {\path{arXiv:0808.2223}},
  \href {https://doi.org/10.4310/ATMP.2011.v15.n6.a1}
  {\path{doi:10.4310/ATMP.2011.v15.n6.a1}}.

\bibitem{Esole:2011sm}
M.~Esole, S.-T. Yau, {Small resolutions of SU(5)-models in F-theory}, Adv.
  Theor. Math. Phys. 17~(6) (2013) 1195--1253.
\newblock \href {http://arxiv.org/abs/1107.0733} {\path{arXiv:1107.0733}},
  \href {https://doi.org/10.4310/ATMP.2013.v17.n6.a1}
  {\path{doi:10.4310/ATMP.2013.v17.n6.a1}}.

\bibitem{Grimm:2011fx}
T.~W. Grimm, H.~Hayashi, {F-theory fluxes, Chirality and Chern-Simons
  theories}, JHEP 03 (2012) 027.
\newblock \href {http://arxiv.org/abs/1111.1232} {\path{arXiv:1111.1232}},
  \href {https://doi.org/10.1007/JHEP03(2012)027}
  {\path{doi:10.1007/JHEP03(2012)027}}.

\bibitem{Hayashi:2013lra}
H.~Hayashi, C.~Lawrie, S.~Schafer-Nameki, {Phases, Flops and F-theory: SU(5)
  Gauge Theories}, JHEP 10 (2013) 046.
\newblock \href {http://arxiv.org/abs/1304.1678} {\path{arXiv:1304.1678}},
  \href {https://doi.org/10.1007/JHEP10(2013)046}
  {\path{doi:10.1007/JHEP10(2013)046}}.

\bibitem{Esole:2015xfa}
M.~Esole, S.-H. Shao, {M-theory on Elliptic Calabi-Yau Threefolds and 6d
  Anomalies} (4 2015).
\newblock \href {http://arxiv.org/abs/1504.01387} {\path{arXiv:1504.01387}}.

\bibitem{KodairaII}
K.~Kodaira, \href{http://www.jstor.org/stable/1970131}{On compact analytic
  surfaces, {II}}, Annals of Mathematics 77~(3) (1963) 563--626.
\newline\urlprefix\url{http://www.jstor.org/stable/1970131}

\bibitem{KodairaIII}
K.~Kodaira, \href{http://www.jstor.org/stable/1970500}{On compact analytic
  surfaces, {III}}, Annals of Mathematics 78~(1) (1963) 1--40.
\newline\urlprefix\url{http://www.jstor.org/stable/1970500}

\bibitem{NeronClassification}
A.~N{\'e}ron, \href{https://doi.org/10.1007/BF02684271}{Mod{\`e}les minimaux
  des vari{\'e}t{\'e}s ab{\'e}liennes sur les corps locaux et globaux},
  Publications Math{\'e}matiques de l'Institut des Hautes {\'E}tudes
  Scientifiques 21~(1) (1964) 5--125.
\newblock \href {https://doi.org/10.1007/BF02684271}
  {\path{doi:10.1007/BF02684271}}.
\newline\urlprefix\url{https://doi.org/10.1007/BF02684271}

\bibitem{Bershadsky:1996nh}
M.~Bershadsky, K.~A. Intriligator, S.~Kachru, D.~R. Morrison, V.~Sadov,
  C.~Vafa, {Geometric singularities and enhanced gauge symmetries}, Nucl. Phys.
  B 481 (1996) 215--252.
\newblock \href {http://arxiv.org/abs/hep-th/9605200}
  {\path{arXiv:hep-th/9605200}}, \href
  {https://doi.org/10.1016/S0550-3213(96)90131-5}
  {\path{doi:10.1016/S0550-3213(96)90131-5}}.

\bibitem{Huang:2018gpl}
Y.-C. Huang, W.~Taylor, {Comparing elliptic and toric hypersurface Calabi-Yau
  threefolds at large Hodge numbers}, JHEP 02 (2019) 087.
\newblock \href {http://arxiv.org/abs/1805.05907} {\path{arXiv:1805.05907}},
  \href {https://doi.org/10.1007/JHEP02(2019)087}
  {\path{doi:10.1007/JHEP02(2019)087}}.

\bibitem{wazir2004arithmetic}
R.~Wazir, Arithmetic on elliptic threefolds, Compositio Mathematica 140~(3)
  (2004) 567--580.

\bibitem{wilson1999flops}
P.~Wilson, Flops, type iii contractions and gromov-witten invariants on
  calabi-yau threefolds, New trends in algebraic geometry 264 (1999) 465.

\bibitem{mcduff1994j}
D.~W. McDuff, D.~Salamon, et~al., $ J $-Holomorphic Curves and Quantum
  Cohomology, no.~6, American Mathematical Soc., 1994.

\bibitem{cadavid199511}
A.~Cadavid, A.~Ceresole, R.~D'Auria, S.~Ferrara, 11-dimensional supergravity
  compactified on calabi-yau threefolds, Physics Letters B 357~(1-2) (1995)
  76--80.

\bibitem{Antoniadis:1995vz}
I.~Antoniadis, S.~Ferrara, T.~R. Taylor, {N=2 heterotic superstring and its
  dual theory in five-dimensions}, Nucl. Phys. B 460 (1996) 489--505.
\newblock \href {http://arxiv.org/abs/hep-th/9511108}
  {\path{arXiv:hep-th/9511108}}, \href
  {https://doi.org/10.1016/0550-3213(95)00659-1}
  {\path{doi:10.1016/0550-3213(95)00659-1}}.

\bibitem{Ferrara:1996hh}
S.~Ferrara, R.~R. Khuri, R.~Minasian, {M theory on a Calabi-Yau manifold},
  Phys. Lett. B 375 (1996) 81--88.
\newblock \href {http://arxiv.org/abs/hep-th/9602102}
  {\path{arXiv:hep-th/9602102}}, \href
  {https://doi.org/10.1016/0370-2693(96)00270-5}
  {\path{doi:10.1016/0370-2693(96)00270-5}}.

\bibitem{Witten:1996qb}
E.~Witten, {Phase transitions in M theory and F theory}, Nucl. Phys. B 471
  (1996) 195--216.
\newblock \href {http://arxiv.org/abs/hep-th/9603150}
  {\path{arXiv:hep-th/9603150}}, \href
  {https://doi.org/10.1016/0550-3213(96)00212-X}
  {\path{doi:10.1016/0550-3213(96)00212-X}}.

\bibitem{Seiberg:1996bd}
N.~Seiberg, {Five-dimensional SUSY field theories, nontrivial fixed points and
  string dynamics}, Phys. Lett. B 388 (1996) 753--760.
\newblock \href {http://arxiv.org/abs/hep-th/9608111}
  {\path{arXiv:hep-th/9608111}}, \href
  {https://doi.org/10.1016/S0370-2693(96)01215-4}
  {\path{doi:10.1016/S0370-2693(96)01215-4}}.

\bibitem{Intriligator:1997pq}
K.~A. Intriligator, D.~R. Morrison, N.~Seiberg, {Five-dimensional
  supersymmetric gauge theories and degenerations of Calabi-Yau spaces}, Nucl.
  Phys. B 497 (1997) 56--100.
\newblock \href {http://arxiv.org/abs/hep-th/9702198}
  {\path{arXiv:hep-th/9702198}}, \href
  {https://doi.org/10.1016/S0550-3213(97)00279-4}
  {\path{doi:10.1016/S0550-3213(97)00279-4}}.

\bibitem{Grimm:2013oga}
T.~W. Grimm, A.~Kapfer, J.~Keitel, {Effective action of 6D F-Theory with U(1)
  factors: Rational sections make Chern-Simons terms jump}, JHEP 07 (2013) 115.
\newblock \href {http://arxiv.org/abs/1305.1929} {\path{arXiv:1305.1929}},
  \href {https://doi.org/10.1007/JHEP07(2013)115}
  {\path{doi:10.1007/JHEP07(2013)115}}.

\bibitem{Bonetti:2013ela}
F.~Bonetti, T.~W. Grimm, S.~Hohenegger, {One-loop Chern-Simons terms in five
  dimensions}, JHEP 07 (2013) 043.
\newblock \href {http://arxiv.org/abs/1302.2918} {\path{arXiv:1302.2918}},
  \href {https://doi.org/10.1007/JHEP07(2013)043}
  {\path{doi:10.1007/JHEP07(2013)043}}.

\bibitem{Witten:1996md}
E.~Witten, {On flux quantization in M theory and the effective action}, J.
  Geom. Phys. 22 (1997) 1--13.
\newblock \href {http://arxiv.org/abs/hep-th/9609122}
  {\path{arXiv:hep-th/9609122}}, \href
  {https://doi.org/10.1016/S0393-0440(96)00042-3}
  {\path{doi:10.1016/S0393-0440(96)00042-3}}.

\bibitem{Grimm:2011sk}
T.~W. Grimm, R.~Savelli, {Gravitational Instantons and Fluxes from M/F-theory
  on Calabi-Yau fourfolds}, Phys. Rev. D 85 (2012) 026003.
\newblock \href {http://arxiv.org/abs/1109.3191} {\path{arXiv:1109.3191}},
  \href {https://doi.org/10.1103/PhysRevD.85.026003}
  {\path{doi:10.1103/PhysRevD.85.026003}}.

\bibitem{Cvetic:2012xn}
M.~Cvetic, T.~W. Grimm, D.~Klevers, {Anomaly Cancellation And Abelian Gauge
  Symmetries In F-theory}, JHEP 02 (2013) 101.
\newblock \href {http://arxiv.org/abs/1210.6034} {\path{arXiv:1210.6034}},
  \href {https://doi.org/10.1007/JHEP02(2013)101}
  {\path{doi:10.1007/JHEP02(2013)101}}.

\bibitem{Dasgupta:1999ss}
K.~Dasgupta, G.~Rajesh, S.~Sethi, {M theory, orientifolds and G - flux}, JHEP
  08 (1999) 023.
\newblock \href {http://arxiv.org/abs/hep-th/9908088}
  {\path{arXiv:hep-th/9908088}}, \href
  {https://doi.org/10.1088/1126-6708/1999/08/023}
  {\path{doi:10.1088/1126-6708/1999/08/023}}.

\bibitem{Donagi:2008ca}
R.~Donagi, M.~Wijnholt, {Model Building with F-Theory}, Adv. Theor. Math. Phys.
  15~(5) (2011) 1237--1317.
\newblock \href {http://arxiv.org/abs/0802.2969} {\path{arXiv:0802.2969}},
  \href {https://doi.org/10.4310/ATMP.2011.v15.n5.a2}
  {\path{doi:10.4310/ATMP.2011.v15.n5.a2}}.

\bibitem{Braun:2011zm}
A.~P. Braun, A.~Collinucci, R.~Valandro, {G-flux in F-theory and algebraic
  cycles}, Nucl. Phys. B 856 (2012) 129--179.
\newblock \href {http://arxiv.org/abs/1107.5337} {\path{arXiv:1107.5337}},
  \href {https://doi.org/10.1016/j.nuclphysb.2011.10.034}
  {\path{doi:10.1016/j.nuclphysb.2011.10.034}}.

\bibitem{Marsano:2011hv}
J.~Marsano, S.~Schafer-Nameki, {Yukawas, G-flux, and Spectral Covers from
  Resolved Calabi-Yau's}, JHEP 11 (2011) 098.
\newblock \href {http://arxiv.org/abs/1108.1794} {\path{arXiv:1108.1794}},
  \href {https://doi.org/10.1007/JHEP11(2011)098}
  {\path{doi:10.1007/JHEP11(2011)098}}.

\bibitem{Krause:2011xj}
S.~Krause, C.~Mayrhofer, T.~Weigand, {$G_4$ flux, chiral matter and singularity
  resolution in F-theory compactifications}, Nucl. Phys. B 858 (2012) 1--47.
\newblock \href {http://arxiv.org/abs/1109.3454} {\path{arXiv:1109.3454}},
  \href {https://doi.org/10.1016/j.nuclphysb.2011.12.013}
  {\path{doi:10.1016/j.nuclphysb.2011.12.013}}.

\bibitem{Beasley:2008kw}
C.~Beasley, J.~J. Heckman, C.~Vafa, {GUTs and Exceptional Branes in F-theory -
  II: Experimental Predictions}, JHEP 01 (2009) 059.
\newblock \href {http://arxiv.org/abs/0806.0102} {\path{arXiv:0806.0102}},
  \href {https://doi.org/10.1088/1126-6708/2009/01/059}
  {\path{doi:10.1088/1126-6708/2009/01/059}}.

\bibitem{Hayashi:2014kca}
H.~Hayashi, C.~Lawrie, D.~R. Morrison, S.~Schafer-Nameki, {Box Graphs and
  Singular Fibers}, JHEP 05 (2014) 048.
\newblock \href {http://arxiv.org/abs/1402.2653} {\path{arXiv:1402.2653}},
  \href {https://doi.org/10.1007/JHEP05(2014)048}
  {\path{doi:10.1007/JHEP05(2014)048}}.

\bibitem{Esole:2015sta}
M.~Esole, S.~G. Jackson, R.~Jagadeesan, A.~G. No\"el, {Incidence Geometry in a
  Weyl Chamber I: $GL_n$} (8 2015).
\newblock \href {http://arxiv.org/abs/1508.03038} {\path{arXiv:1508.03038}}.

\bibitem{Esole:2016npg}
M.~Esole, S.~G. Jackson, R.~Jagadeesan, A.~G. No\"el, {Incidence Geometry in a
  Weyl Chamber II: $SL_n$} (1 2016).
\newblock \href {http://arxiv.org/abs/1601.05070} {\path{arXiv:1601.05070}}.

\bibitem{Hayashi:2009ge}
H.~Hayashi, T.~Kawano, R.~Tatar, T.~Watari, {Codimension-3 Singularities and
  Yukawa Couplings in F-theory}, Nucl. Phys. B 823 (2009) 47--115.
\newblock \href {http://arxiv.org/abs/0901.4941} {\path{arXiv:0901.4941}},
  \href {https://doi.org/10.1016/j.nuclphysb.2009.07.021}
  {\path{doi:10.1016/j.nuclphysb.2009.07.021}}.

\bibitem{achmed2018note}
I.~Achmed-Zade, I.~Garc{\'\i}a-Etxebarria, C.~Mayrhofer, A note on non-flat
  points in the $su(5) \times u (1) $ pq f-theory model, arXiv preprint
  arXiv:1806.05612 (2018).

\bibitem{JLT}
P.~Jefferson, S.-Y. Li, W.~Taylor, (in progress).

\bibitem{article}
A.~Mustata, The structure of a local embedding and chern classes of weighted
  blow-ups (11 2012).

\end{thebibliography}

\end{document}